\documentclass[envcountsame,envcountreset,envcountsect]{llncs}

\usepackage{multirow}

\usepackage{hyperref}

\newcommand{\pdiam}{{player~$\Diamond$}\xspace}
\newcommand{\pbox}{{player~$\Box$}\xspace}

\usepackage{xspace}
\usepackage{amsmath,amssymb}
\usepackage{stmaryrd} 

\usepackage[appendix=inline]{apxproof} 
\usepackage{todonotes} 

\usepackage{tikz-cd}

\usepackage{tikz}
 \usetikzlibrary{trees}
 \usetikzlibrary{shapes}
 \usetikzlibrary{fit}
 \usetikzlibrary{shadows}
 \usetikzlibrary{backgrounds}
 \usetikzlibrary{arrows,automata}

\tikzset{
   n/.style= {circle,fill,inner sep=1.5pt,node distance=2cm}
  ,acc/.style={circle,draw,inner sep=3pt,node distance=2cm}
  ,phantom/.style={circle},
  ,arr/.style={->, >=stealth, semithick, shorten <= 3pt, shorten >= 3pt}
}

\spnewtheorem{defn}[theorem]{Definition}{\bfseries}{\rmfamily}

\renewcommand{\Box}{\square}
\renewcommand{\Diamond}{\lozenge}
\newcommand{\pow}{\mathcal{P}}
\newcommand{\sem}[1]{\llbracket #1 \rrbracket}

\newcommand{\FL}{\mathsf{FL}}

\newcommand{\At}{\mathsf{At}}

\newcommand{\Var}{\mathsf{Var}}

\newcommand{\takeout}[1]{\empty}

\newcommand\ExpTime{$\textsc{ExpTime}$\xspace}

\newcommand{\paritytime}[2]{$\textsc{parity}({#1},{#2})\textsc{-Time}$\xspace}

\newcommand\buchi{B\"uchi~}
\newcommand\cbuchi{Co-B\"uchi~}

\begin{document}

\title{A Survey on Satisfiability Checking
for the $\mu$-Calculus through Tree Automata\thanks{This work is
supported by the ERC Consolidator grant D-SynMA (No. 772459).}}
\titlerunning{Surveying $\mu$-Calculus Satisfiability through
	Tree Automata} 
\author{%
  Daniel Hausmann and Nir Piterman
}
\authorrunning{D.~Hausmann and N.~Piterman}

\institute{Gothenburg University, Gothenburg, Sweden
}

\maketitle
\begin{abstract}
	Algorithms for model checking and satisfiability of the
	modal $\mu$-calculus start by converting formulas to
	alternating parity tree automata.
	Thus, model checking is reduced to checking acceptance by
	tree automata and satisfiability to checking their emptiness.
	The first reduces directly to the solution of parity games
	but the second is more complicated.

	We review the non-emptiness checking of alternating tree
	automata by a reduction to
	solving parity games of a certain structure, so-called
	\emph{emptiness games}.
	Since the emptiness problem for alternating tree automata is
	\ExpTime-complete, the size of these games is exponential in
	the number of states of the input automaton.
	We show how the construction of the emptiness games combines
	a (fixed) structural part with \mbox{(history-)}determinization of
	parity word automata.
	For tree automata with certain syntactic structures,
	simpler methods may be used to handle the treatment of the
	word automata, which then may be asymptotically smaller than
	in the general case.

	These results have direct consequences in satisfiability and
	validity checking
	for (various fragments of) the modal $\mu$-calculus.
\end{abstract}
%

\section{Introduction}
\label{section:introduction}

The modal $\mu$-calculus extends modal logic with least and greatest
fixpoint operators \cite{Kozen83}.
The $\mu$-calculus is expressive enough to express many temporal
logics, in particular it can capture CTL$^*$ and its fragments LTL and
CTL \cite{EmersonL86}.
At the same time, the $\mu$-calculus has interesting algorithmic
and algebraic properties.
For example, the $\mu$-calculus model-checking problem is equivalent to
the solution of parity games, a well known (still) open problem
attracting much research.
This combination led to high interest in the $\mu$-calculus, studying
many aspects of the logic.

Here, we are interested in the question of satisfiability of the
$\mu$-calculus.
The problem is \ExpTime-complete and the first
algorithms of this complexity were automata based \cite{EmersonJutla99}.
Much like for other temporal logics, initial treatment of the logic was
done through nondeterministic automata \cite{StreettE89}.
However, later, the richer structure of alternating automata enabled
translations that are more natural and direct \cite{MullerSS88}.
This translation defers the complicated handling of the satisfiability
problem to standard automata constructions.

With this approach,
the satisfiability problem for the modal $\mu$-calculus reduces to the
non-emptiness problem of alternating parity tree automata
\cite{MullerSS88}.
The latter problem is solved either by constructing equivalent
nondeterministic parity tree automata \cite{MullerS87} or by a direct
reduction to two-player perfect information parity games \cite{Wilke01}.

We revisit the reduction and present it in a way that separates the
tree acceptance and the parity acceptance
aspects of a parity tree automaton $A$.
The method creates an arena $G_A$ (\emph{strategy arena}), which
captures all the decisions made at the same location by $A$ and a
nondeterministic parity word automaton $T_A$ (\emph{tracking
automaton}) that ``accepts'' bad branches
in run-trees of the original automaton.
The original automaton then is non-empty if and only
if the combination of $G_A$ with $T_A$ as losing condition is
won by the existential player.
By using a history deterministic (or fully deterministic) word automaton $H_A$ that accepts the same language as
$T_A$, we construct
a parity game $G^*_A$.
\begin{center}
\begin{tikzcd}
&& A \text{ (APT)}
\arrow[ddl, "\text{powerset-like
construction}"']\arrow[dr,"\text{structure preservation}"]\\
& & &
T_A\text{ (NPW)}
\arrow[d,"\text{history determinization / determinization}"']{a}\\
&G_A \text{ (game arena)} \arrow[dr] & & H_A \text{ (HD-PW)}
\arrow[dl]\\
&& G^*_A \text{ (parity game)}
\end{tikzcd}
\end{center}
This approach reduces the algorithmic content of non-emptiness checking
for alternating automata (and for satisfiability checking in the modal
$\mu$-calculus) to a fixed construction of a game arena and
(history-)determinization of word automata that depend on the exact
structure of the original automaton ($\mu$-calculus formula).

We then show that as the structure of $T_A$ strongly depends on the
structure of $A$, specialized history determinization and
determinization constructions lead to complexity results that match
bespoke algorithms for different fragments of the $\mu$-calculus.
These results are summarized in the following table, where LL stands
for limit linear, LD for limit deterministic, N for nondeterministic,
HD for history deterministic, D for deterministic, W for weak, B for B\"uchi, C for
co-B\"uchi, and P for parity.
For example, LL-CW is a limit linear co-B\"uchi word automaton
and LD-WT is a limit deterministic weak tree automaton.

\begin{center}
\begin{tabular}{|c |c |r| l| r| c |}

\hline
\multicolumn{2}{|c|}{}
& type of $T_A$ & method & type of $H_A$ & size of $H_A$\\
\hline
\parbox[t]{2mm}{\multirow{4}{*}{\rotatebox[origin=c]{90}{\quad\,\,\scalebox{0.8}{\cbuchi}}}}
   & \multicolumn{1}{c|}{\multirow{2}{*}{det.}}
& LL-CW & circle method & DCW & $n^2\cdot 2^n$ \\ \cline{3-6}
&& NCW & Miyano-Hayashi & DCW & $3^n$ \\ \cline{2-6}
& {history-det.} & LD-CW & focus method & HD-CW & $n\cdot 2^n$ \\
\hline

\parbox[t]{2mm}{\multirow{4}{*}{\rotatebox[origin=c]{90}{\quad\,\,\scalebox{0.8}{\buchi}}}}   & \multicolumn{1}{c|}{\multirow{2}{*}{det.}}
& LD-BW & permutation method & DPW & $\mathcal{O}(n!)$  \\ \cline{3-6}
&& NBW & Safra-Piterman & DPW & $\mathcal{O}((n!)^2)$ \\ \cline{2-6}
& history-det. & NBW & Henzinger-Piterman & HD-PW & $\mathcal{O}(3^{n^2})$ \\
\hline

\end{tabular}
\end{center}

Going back to the $\mu$-calculus, the following table depicts
relations between various syntactic properties of $\mu$-calculus
formulas (formally defined later) and automata.
Thus, the separated treatment of the arena and the acceptance, and the
different constructions for word automata summarize in one framework
complexity results relating to various different fragments of the
$\mu$-calculus.

\begin{center}
\begin{tabular}{|c |c |c |}

\hline
property of $\varphi$ & type of $A$ & type of $T_A$\\
\hline
limit-linear & LL-WT & LL-CW\\
alternation-free & AWT & NCW\\
aconjunctive alternation-free & LD-WT & LD-CW\\
aconjunctive & LD-PT & LD-BW\\
unrestricted & APT & NBW\\
\hline

\end{tabular}
\end{center}

\section{The Modal $\mu$-Calculus}
\label{section:mu-calculus}

We are concerned with satisfiability checking for different syntactical fragments of
the branching-time $\mu$-calculus, introduced by Kozen~\cite{Kozen83}.

\paragraph{Syntax.} Formulas of the $\mu$-calculus are generated by the following grammar,
where $\mathsf{At}$ and $\mathsf{Var}$ are countable sets of
atoms and fixpoint variables, respectively:
\begin{align*}
\varphi,\psi:= p\mid \neg p \mid \varphi\wedge\psi\mid \varphi\vee\psi \mid \Diamond \varphi\mid\Box\varphi
\mid X\mid \xi X.\,\varphi \quad {p\in\At, X\in\Var,\xi\in\{\mu,\nu\}}
\end{align*}
Fixpoint operators $\mu X.$ and $\nu X.$ \emph{bind} their variable $X$,
giving rise to standard notions of \emph{bound} and \emph{free variables};
for $\mu X.\,\psi$, free occurrences of $X$ in $\psi$ are \emph{least fixpoint variables} and
for $\nu X.\,\psi$, free occurrences of $X$ in $\psi$ are \emph{greatest fixpoint variables}.

The \emph{Fischer-Ladner closure}~\cite{Kozen83} $\FL(\varphi)$ (or just \emph{closure}) of a
closed formula $\varphi$
is the least set of formulas that contains $\varphi$ and is closed under taking subformulas for non-fixpoint
operators and under unfolding for fixpoint operators; e.g.
$\mu X.\,\psi\in\FL(\varphi)$ implies
$\psi[X\mapsto\mu X.\,\psi]\in\FL(\varphi)$, where $\psi[X\mapsto\mu
X.\,\psi]$
is the formula that
is obtained from $\psi$ by replacing every free occurrence of the variable $X$ in $\psi$ by
the formula $\mu X.\,\psi$. We have $\FL(\varphi)\leq |\varphi|$ where $|\varphi|$ is the number of operators
that are required to write $\varphi$ (that is, the number of nodes in the syntax tree of $\varphi$).

A formula $\varphi$ is \emph{clean}, if all fixpoint variables
are bound at most once in it. Then we denote \emph{the} fixpoint formula $\xi X.\,\psi$
that binds a variable $X$ in $\varphi$ by $\theta_\varphi(X)$.
While transforming an arbitrary formula to a clean formula
(by renaming bound variables accordingly) can increase the closure size, a
translation to tree automata that does not rely on cleanness has recently been given in~\cite{KupkeMartiVenema22}. For brevity of presentation we assume throughout
that target formulas are clean but remark that this does not affect the stated complexity
results since the more involved translation from~\cite{KupkeMartiVenema22}
can be used to obtain tree automata (of suitable size and rank) from
arbitrary formulas.

\begin{remark}
Another common constraint on the syntactic structure of formulas is
\emph{guardedness}, which requires that there is always at least one
modal operator between a fixpoint operator and occurrences of the
fixpoint variable that it binds. It is currently an open question whether
there is a guardedness-transformation with polynomial blow-up of the
closure size, and it has been shown that such a polynomial transformation
would yield a polynomial algorithm for parity game
solving~\cite{KupkeMartiVenema22}. Throughout this work, we assume
that formulas are guarded.
\end{remark}

\paragraph{Alternation-depth.}
Given a clean formula $\varphi$, and two subformulas $\xi X.\,\psi$ and
$\xi' Y.\,\chi$ of $\varphi$,
$\xi' Y.\,\chi$ \emph{depends} on $\xi X.\,\psi$ if $X$ has a free occurrence in
$\chi$. We define the \emph{dependent nesting order} $\succeq_\varphi$
to be the partial order on fixpoint subformulas of $\varphi$ obtained by taking the
reflexive-transitive closure of the dependency ordering.
The \emph{alternation-depth} $\mathsf{ad}(\varphi)$ of $\varphi$ then is defined
to be the maximal length of an alternating $\succ_\varphi$-path, where
a $\succ_\varphi$-path is alternating if all its transitions switch the fixpoint type.
Given a fixpoint subformula $\chi=\eta X.\,\psi$ of $\varphi$, let $d$ be
the maximal length of an alternating $\succeq_\varphi$-path
that starts at $\chi$.
We define the \emph{alternation-level} $\mathsf{al}(\chi)$
of $\chi$ to be $2\lceil d/2\rceil-1$ if
$\eta=\mu$ and
$2\lfloor d/2\rfloor$ if $\eta=\nu$.
Hence least fixpoint formulas have odd alternation-level while greatest fixpoint
formulas have even alternation-level; furthermore, we always have
$\mathsf{al}(\chi)\leq \mathsf{ad}(\varphi)$.

\paragraph{Semantics.}
Formulas of the $\mu$-calculus are evaluated over pointed
Kripke structures.
A pointed Kripke structure is $K=(W,w_0,R,L)$, where $W$ is a set
of worlds, $w_0\in W$ is an initial world, $R\subseteq W\times
W$ is a transition relation, and $L:W \rightarrow
\pow(\mathsf{At})$ is a labeling function.
We denote by $R(w)=\{w' \in W \mid (w,w')\in R\}$ the set of
worlds connected by $R$ to $w$.
We restrict attention to structures where for every $w\in W$ we
have $R(w)\neq \emptyset$.

Given a $\mu$-calculus formula $\varphi$ and a pointed Kripke
structure $K$, the semantics of the formula is defined based on a
valuation function $\eta$ assigning each variable appearing in
$\varphi$ to a set of worlds of $K$.
Given such a function $\eta$ we denote by $\eta[X\leftarrow
S]$ the function $\eta'$ where $\eta'(X)=S$ and
$\eta'(Y)=\eta(Y)$ for every $Y\neq X$.
The semantics of the $\mu$-calculus is included in Figure~\ref{fig:mu
calculus semantics}.
\begin{figure}[bt]
\begin{align*}
	\sem{p}_\eta &= \{w ~|~ p\in L(w)\}\\
	\sem{\neg p}_\eta &= \{w ~|~ p\notin L(w)\}\\
	\sem{\varphi\vee\psi}_\eta &= \sem{\varphi}_\eta \cup
\sem{\psi}_\eta\\
	\sem{\varphi\wedge\psi}_\eta &= \sem{\varphi}_\eta \cap
\sem{\psi}_\eta\\
 \sem{\Diamond \varphi}_\eta &= \{ w ~|~ R(w) \cap
 \sem{\varphi}_\eta \neq \emptyset\}\\
 \sem{\Box \varphi}_\eta &= \{ w~|~ R(w) \subseteq
 \sem{\varphi}_\eta\}\\
 \sem{X}_\eta &= \eta(X) \\
 \sem{\mu X.\varphi}_\eta &= \bigcap \{T\subseteq W ~|~
 \sem{\varphi}_{\eta[X\leftarrow T]} \subseteq T\}\\
 \sem{\nu X.\varphi}_\eta &= \bigcup \{T\subseteq W ~|~ T
 \subseteq \sem{\varphi}_{\eta[X\leftarrow T]}\}
\end{align*}
\vspace*{-20pt}
\caption{Semantics of the $\mu$-calculus}
\label{fig:mu calculus semantics}
\vspace*{-20pt}
\end{figure}
It is simple to see that in the case that a formula $\varphi$
is closed its semantics does not depend on the initial
valuation $\eta$. Thus, for a closed formula we write
$\sem{\varphi}$.
A formula $\varphi$ is satisfiable if there exists a structure
$K=(W,w_0,R,L)$ such that $w_0\in \sem{\varphi}$.

\begin{theorem}[\hspace{-0.2pt}\cite{EmersonJutla99}]
Given a $\mu$-calculus formula $\varphi$, deciding whether
$\varphi$ is satisfiable is \ExpTime-complete.
\end{theorem}

\paragraph{Fragments of the $\mu$-calculus.}
We consider the following fragments of the $\mu$-calculus:
\begin{itemize}
	\item
	The \emph{limit-linear} fragment of the $\mu$-calculus consists of
	all formulas $\varphi$ such that for all subformulas $\mu X.\,\psi$ of $\varphi$,
    $X$ has exactly one occurrence in $\psi$, and this occurrence is
    not in the scope of a fixpoint subformula of $\psi$. Computation tree logic (CTL)
    is a fragment of the limit-linear $\mu$-calculus.
	\item
	The \emph{alternation-free} fragment of the $\mu$-calculus consists
	of all formulas $\varphi$ such that $\mathsf{ad}(\varphi)\leq 1$.
	It has been shown that satisfiability checking for a guarded alternation-free formula
	of size $n$ can be done by solving a \buchi game of size
	$3^n$~\cite{FriedmannLatteLange13}.
	\item
	The \emph{aconjunctive} fragment of the $\mu$-calculus consists of
	all formulas $\varphi$ such that for all conjunctions
	$\psi\wedge\chi$ that occur as a
	subformula in $\varphi$, at most one of the conjuncts $\psi$ or
	$\chi$ contains
	a free least fixpoint variable.
	Satisfiability checking for a (weakly) aconjunctive formula
	of size $n$ and with $k$ priorities can be done by solving a parity game of size
	$e\cdot(nk)!$ and with $2nk$ priorities~\cite{HausmannEA18}.
	\item
	The \emph{alternation-free aconjunctive} fragment of the $\mu$-calculus is the
	intersection of the alternation-free fragment and the aconjunctive fragment.
	In particular, every limit-linear formula is alternation-free and aconjunctive.
\end{itemize}

\section{Two-Player Games and Alternating Parity Tree Automata}
\label{section:tree automata}
We give background on two-player games and tree automata.
Our notations are based on those developed by
Wilke \cite{Wilke01}.

\begin{defn}
	A \emph{game} is $G=(V,V_\Diamond,V_\Box,E,\alpha)$, where
	$V$ is a set of nodes, $V_\Diamond$ and $V_\Box$ form a
	partition of $V$ to \pdiam and \pbox
	nodes, $E\subseteq V\times V$ is a set of edges, and
	$\alpha\subseteq V^\omega$
	is a winning condition.
	A \emph{play} is a sequence $\pi=v_0,v_1,\ldots$ such that
	for every $i$ we have $(v_i,v_{i+1})\in E$.
	A play $\pi$ is winning for \pdiam if
	$\pi\in\alpha$.
	A strategy for \pdiam is $\sigma:V^*\cdot
	V_\Diamond\rightarrow V$ such that $(v,\sigma(wv))\in E$
	for $wv\in V^*\cdot	V_\Diamond$.
	A play $\pi$ is \emph{compatible} with $\sigma$ if whenever
	$v_i\in V_\Diamond$ we have $v_{i+1}=\sigma(v_0,\ldots,
	v_i)$.
	A strategy for \pdiam is winning from node $v$ if all
	the plays starting in $v$ that are compatible with $\sigma$
	are winning for her.
	Strategies for \pbox are defined similarly.
\end{defn}

	In a \emph{parity} game, there exists a priority function
	$\Omega:V\to\mathbb{N}$ and a play $\pi$ is winning for
	\pdiam if the maximal priority occurring infinitely
	often in $\pi$ is even. A \emph{\buchi} game is a parity
	game with just the priorities $1$ and $2$.
	Given an infinite sequence $\pi \in V^\omega$ let $\mathsf{inf}(\pi)$
	denote the set of
	nodes occuring in $\pi$ infinitely often and put
	$\mathsf{inf}_\Omega(\pi)=\{\Omega(v) ~|~ v\in
	\mathsf{inf}(\pi)\}$.
	Then the parity winning condition induced by $\Omega$ is
    $\alpha = \{\pi\in V^\omega ~|~ \max(\mathsf{inf}_\Omega(\pi)) \mbox{ is
	even}\}$.
	The complexity of analyzing parity games is a hot area of research
	\cite{CaludeEA17,ColcombetF19,CzerwinskiDFJLP19}.
	Here we denote by \paritytime{n}{k} the time complexity of solving
	parity
	games with $n$ nodes and $k$ priorities.
	We do not refer to space complexity, however, a similar general
	dependency on space can be stated.
	Our results produce parity games of different parameters depending on
	the exact shape of a $\mu$-calculus formula.
	We hence use this parametric form to give complexity results.

	\begin{theorem}[Parity and \buchi games
	\cite{CaludeEA17,ChatterjeeH12}]
	Parity games with $n$ nodes and $k$ priorities can be solved in time
	quasipolynomial
	in $n$ and $k$, more specifically\footnote{This improved bound has been shown
	 in~\cite{JurdzinskiMorvan20}} in time
	$n^{2\log(k/\log n)+\mathcal{O}(1)}$, and in polynomial time if $k<\log n$.
	\buchi~games with $n$ nodes can be solved in time $\mathcal{O}(n^2)$.
	\end{theorem}

	We sometimes consider games where edges are labeled.
	In such a case, there exists a set of labels $D$ and we have
	$E\subseteq V\times D \times V$.
	Then, $\alpha$ can be a subset of $V\cdot (D\cdot
	V)^\omega$.

	When the winning condition is not important, we call
	$(V,V_\Diamond,V_\Box,E)$ an \emph{arena}.

\begin{defn}
An \emph{alternating parity tree automaton} $A=(\Sigma,
Q,q_0,\delta,\Omega)$
consists of a finite alphabet $\Sigma$, a finite set of states $Q$, an
initial state $q_0\in Q$, a
transition function $\delta: Q \times \Sigma \to
\mathcal{P}(Q)$, and a priority
function $\Omega:Q\to\mathbb{N}$; furthermore, each state
$q\in Q$ is marked as either local-existential,
local-universal, modal-existential or modal-universal
(denoted by $q\in Q_\vee$, $q\in Q_\wedge$, $q\in Q_\Diamond$,
and $q\in Q_\Box$, respectively).
We also denote $Q_l=Q_\vee \cup Q_\wedge$, where $l$ stands
for local, and $Q_m=Q_\Diamond \cup Q_\Box$, where $m$ stands
for modal.
Without loss of generality, we assume that modal-existential
and modal-universal states have exactly
one successor, formally: $|\delta(q)|=1$ if $q\in Q_\Diamond$
or $q\in Q_\Box$.
We overload $\delta(q)$ to denote $q'$ if $\delta(q)=\{q'\}$.
The \emph{rank} of $A$ is the number $|\Omega(Q)|$ of priorities appearing
in its priority function.
We assume that $A$ does not have local loops.
That is, for every letter $\sigma$ and for every sequence of states
$q_1,\ldots, q_l \in Q_l^+$ such
that for every $i\geq 1$ we have $q_{i+1}\in \delta(q_i,\sigma)$ we
have $q_l\neq q_1$.

A tree automaton is \emph{weak} if for all its strongly connected components $C$,
either all states in $C$ have priority $0$ or all states in $C$ have priority $1$.
A weak tree automaton is \emph{limit-linear} if for each $q\in Q$ such that
$\Omega(q)=1$, there is exactly one path from $q$ to $q$.
That is, within rejecting strongly connected components the looping behavior is deterministic.
A tree automaton is \emph{limit-deterministic} if for each odd priority $p$ and
each state $q$ such that $\Omega(q)=p$, we have that for all $a\in \Sigma$ and
all states $q'\in Q_\wedge$
that are reachable from $q$ by visiting nodes with priority at most $p$,
$|\delta(q',a)\cap Q_{\leq p}|\leq 1$, where $Q_{\leq p}=\{q\in Q\mid \Omega(q)\leq p\}$.
In particular, every limit-linear automaton is limit-deterministic.

Alternating tree automata read \emph{pointed Kripke structures}.
An alternating tree automaton $A$ accepts a Kripke structure
$K=(W,w_0,R,L)$ if \pdiam wins the node $(w_0,q_0)$ in the acceptance game $G_{A,K}$.
Formally, $G_{A,K}=(V,V_\Diamond,V_\Box,E,\alpha)$, where
$V=W \times Q$, $V_\Diamond = W \times (Q_\vee \cup
Q_\Diamond)$, $V_\Box = W\times (Q_\wedge \cup Q_\Box)$,
$\alpha$ is induced by the priority function
$\Omega'(w,q)=\Omega(q)$, and $E$ is defined as follows.
$$
\begin{array}{r c l}
	E & = & \{ ((w,q),(w,q')) ~|~ q\in Q_l \mbox{ and } q'\in
	\delta(q,L(w))\} \quad \cup \\
	& & \{((w,q),(w',q')) ~|~ q\in Q_m, q'\in \delta(q,L(w)),
	\mbox { and } w'\in R(w)\}
	\end{array}
$$
An automaton $A$ is \emph{non-empty} if there exists a Kripke
structure that it accepts.
\end{defn}

We now state (the well known result) that given a
$\mu$-calculus formula, we can
construct an alternating tree automaton accepting exactly the
models of the formula.

\begin{defn}[Formula automaton]
Given a closed and clean 
$\mu$-calculus formula $\varphi$
that mentions atoms $\mathsf{A}\subseteq\mathsf{At}$,
we define an alternating parity tree automaton
$A(\varphi)=(\Sigma,Q,q_0,\delta,\Omega)$ by putting
$\Sigma=\pow(\mathsf{A})$, $Q=\FL(\varphi)\cup\{\top,\bot\}$, and
$q_0=\varphi$.
We define a partial priority function
$\Omega':Q\rightharpoonup \{0,\ldots,\mathsf{ad}(\varphi)\}$
by putting $\Omega'(\eta X.\,\psi)=\mathsf{al}(\theta(X))$
(recalling that $\theta(X)$ is \emph{the} subformula of
$\varphi$ that binds $X$),
$\Omega'(\bot) = 1$ and $\Omega'(\top)=0$; then $\Omega'$ assigns
a priority to at least one state on each cycle in $A(\varphi)$.
The total priority function $\Omega:V\to\{0,\ldots,\mathsf{ad}(\varphi)\}$
is obtained by putting, for each state $q\in Q$ such that $\Omega'(q)$ is undefined,
$\Omega(q)=p$ where $p$ is the minimum priority such that all paths from $q$ to $q$ visit
priority at most $p$; states that do not belong to a strongly connected component
obtain priority $0$. Furthermore, we put
\begin{align*}
\delta(q,P)&=
\begin{cases}
\{\psi_0,\psi_1\}
& \text{if } q=\psi_0\wedge\psi_1 \text{ or } q=\psi_0\vee\psi_1\\
\{\psi\}
& \text{if } q=\Diamond\psi\text{ or }q=\Box\psi\\
\{\psi[X\mapsto\eta X.\,\psi]\}
& \text{if } q=\eta X. \,\psi\\
\{\top\}
& \text{if } q=p\text{ and }p\in P\text{ or }q=\neg p\text{ and }p\notin P\\
\{\bot\}
& \text{if } q=p\text{ and }p\notin P\text{ or }q=\neg p\text{ and }p\in P\\
\{q\} & \text{if } q=\top\text{ or }q=\bot
\end{cases}
\end{align*}
for $q\in Q$, $P\in\Sigma$.
Finaly, we put
\begin{align*}
Q_\exists &= \{\psi_0\vee\psi_1, \eta X.\,\psi\in \FL(\varphi)\}\cup\{\bot\}&
Q_\Diamond &= \{\Diamond\psi\in \FL(\varphi)\}\\
Q_\forall &= \{\psi_0\wedge\psi_1, p, \neg p\in \FL(\varphi)\}\cup\{\top\} &
Q_\Box &= \{\Box\psi\in \FL(\varphi)\}.
\end{align*}
\end{defn}

\begin{theorem}[\hspace{-0.2pt}\cite{Wilke01,KupkeMartiVenema22}]
We have $L(A(\varphi))=\{(W,w_0,R,L)\mid
w_0\in\sem{\varphi}\}$.
Furthermore, $|Q|\leq |\FL(\varphi)|+2$ and $A(\varphi)$ has rank
$\mathsf{ad}(\varphi)+1$.
\label{theorem:mucalculus2apt}
\end{theorem}

\begin{corollary}
	Deciding if a Kripke structure with set of worlds $W$ satisfies a $\mu$-calculus formula $\varphi$ is
	in \paritytime{|W|\cdot (\FL(\varphi)+2)}{\mathsf{ad}(\varphi)+1}.
\end{corollary}

It follows from Theorem~\ref{theorem:mucalculus2apt} that by
checking whether the language of $A(\varphi)$ is empty we can
decide whether $\varphi$ is satisfiable.
In the next section, we proceed to show how to determine
whether the language of an automaton is empty.

Before proceeding, we show that in case the $\mu$-calculus
formula has a special structure, as defined in
Section~\ref{section:mu-calculus}, the automaton resulting
from the translation above has also a special structure.

\begin{lemma}
\begin{itemize}
	\item If $\varphi$ is alternation-free, then $A(\varphi)$ is
	a weak tree automaton.
	\item If $\varphi$ is limit linear, then $A(\varphi)$ is
	limit linear.
	\item If $\varphi$ is
	aconjunctive, then $A(\varphi)$ is limit deterministic.
\end{itemize}
\end{lemma}
\begin{proof}
\begin{itemize}
\item Let $\varphi$ be alternation-free so that $\mathsf{ad}(\varphi)\leq 1$
and $\mathsf{al}(\psi)\leq 1$ for all $\psi\in\FL(\varphi)$. Hence
$A(\varphi)$ uses just the priorities $\{0,1\}$. Furthermore,
every strongly connected component in $A(\varphi)$ belongs to
either a greatest or a least fixpoint and hence consists
only of states with priority $0$ or only of states with priority $1$, as claimed.
\item Let $\varphi$ be limit linear so that for all subformulas $\mu X.\,\psi$ of $\varphi$,
$X$ has exactly one occurrence in $\psi$. Then $\varphi$ is alternation-free
so that $A(\phi)$ is weak by the previous item. Furthermore, all states in the
strongly connected component of $\theta(X)$ belong to $\mu X.\,\psi$ and hence have
priority $1$. Since $\varphi$ is limit-linear, there is exactly one circular path
in the strongly connected component of $\theta(X)$. Hence $A(\varphi)$ is limit linear.
\item Let $\varphi$ be aconjunctive, let $p$ be an odd number, let $q\in Q$ such that
$\Omega(q)=p$ and let $q'\in Q_\forall$
be a state that is reachable from $q$ by visiting states with
priority at most $p$. It remains to show that for all $a\in \Sigma$, we have
$|\delta(q',\Sigma)\cap Q_{\leq p}|\leq 1$, where $Q_{\leq p}=\{q\in Q\mid \Omega(q)\leq p\}$.
Since $q\in Q_\wedge$, we have $q=\psi_1\wedge\psi_2$ for some $\psi_1,\psi_2\in\FL(\varphi)$.
As $\varphi$ is aconjunctive, there is at most one $i$ such that $\psi_i$ contains
a free least fixpoint variable, and such that an odd
priority is reachable from $\psi_i$ without first passing a priority greater than $p$.
Hence $\Omega(\psi_1)>p$ or $\Omega(\psi_2)>p$, showing that $|\delta(q',\Sigma)\cap Q_{\leq p}|\leq 1$.
\end{itemize}
\qed
\end{proof}

\section{Emptiness of Alternating Tree Automata}
\label{section:emptiness}
We now show how the decision whether the language of an
alternating automaton is non-empty can be reduced to deciding
the winner in a two-player game.
We start by constructing a game with labeled edges and
an acceptance condition that is defined by a nondeterministic word
automaton.
We show that \pdiam wins in this game if the language of the
alternating automaton is not empty.
Then, by manipulating the word automaton, we construct a parity
game with the same quality: \pdiam wins if the language of the
alternating automaton is not empty.
This is interesting because it unifies many results about
fragments of the $\mu$-calculus to results about word automata.

\subsection{Nondeterministic Parity Word Automata}
Before proceeding we introduce nondeterministic and history
deterministic word automata.

\begin{defn}[Nondeterministic Parity Word Automata]
A nondeterministic parity word automaton is
$N=(\Sigma,Q,q_0,\delta,\Omega)$, where $\Sigma$ is a finite  alphabet,
$Q$ a finite set of states, $q_0\in Q$ an initial state, and
$\delta:Q\times\Sigma\to \pow(Q)$ a transition function.
The \emph{priority function}
$\Omega:Q\to\mathbb{N}$ assigns priorities to states.
Given an automaton $N$, the rank of $N$ is its maximal priority, that is
$\max\{\Omega(q)\mid q\in Q\}$.
Given an infinite word $w=a_0a_1\ldots\in \Sigma^\omega$, a \emph{run
of $N$ on $w$} is an infinite sequence $\tau=q_0,q_1,\ldots$ of states
such that $q_{i+1}\in \delta(q_i,a_i)$ for all $i\geq 0$.
A run $\tau=q_0,q_1,\ldots$ is \emph{accepting} if the highest priority
that occurs infinitely
often in $\tau$ is even.
Formally, reusing the notation $\mathsf{inf}_\Omega$ introduced for
parity games, run $\tau$ is accepting if and only if
$\max\{\mathsf{inf}_\Omega(\tau)\}$ is an even number.
The language accepted by $N$ is
\begin{align*}
	L(N)=\{w\in\Sigma^\omega \mid \text{there is an accepting run of $N$
	on $w$}\}.
\end{align*}
\end{defn}

\begin{defn}[History-deterministic Word
Automata~\cite{HenzingerP06}]
	Given a nondeterministic word automaton $N$,
	a \emph{resolver} for $N$ is a function $\sigma:\Sigma^*\to
	Q$ such that $\sigma(\epsilon)=q_0$ and for all sequences
	$wa\in \Sigma^+$, we have $\sigma(wa)\in \delta(\sigma(w),a)$.
	Given a word $w=a_0a_1\cdots\in \Sigma^\omega$ the outcome of
	$\sigma$ on $w$, denoted $\sigma(w)$, is the run $r=q_0,q_1,\ldots$
	such that for all $i\geq 0$ we have $q_i=\sigma(a_0\ldots
	a_{i-1})$.
	We say that $N$ is
	\emph{history-deterministic} if there is a resolver $\sigma$
	such that for every word $w$
	we have that
	\begin{align*}
		w\in L(N)\text{ if and only if } \sigma(w)\text{ is an accepting
		run of } N.
	\end{align*}
\end{defn}

A word automaton is \emph{deterministic} if for every state $q\in Q$ and
every letter $a\in\Sigma$ we have $|\delta(q,a)|\leq 1$.
In particular, every deterministic automaton is history deterministic.

\begin{theorem}[\hspace{-0.2pt}\cite{Safra88,Piterman07,HenzingerP06}]
	Given a nondeterministic parity word automaton $N$, there exist a
	history deterministic parity automaton $H$ and a deterministic parity
	automaton $D$ such that $L(N)=L(H)=L(D)$.
	\label{theorem:history determinization and determinization}
\end{theorem}

In Section~\ref{section:word automata} we mention several
determinization and history determinization constructions that take
nondeterministic word automata and construct equivalent (history) deterministic
automata.

\subsection{The Emptiness Games}
Using these definitions we are ready to proceed with the
construction of the games capturing emptiness of an
alternating parity tree automaton.

\begin{defn}[Strategy Arena]
Given an alternating parity tree automaton
$A=(\Sigma,Q,q_0,\delta,\Omega)$
we define the \emph{strategy arena}
$G_A=(V,V_\Diamond,V_\Box,E)$, where the components of $G_A$
are as follows. We label the edges in $E$ as we explain below.
\begin{itemize}
	\item
	$V = (\mathcal{P}(Q) \times \Sigma)
	\cup \mathcal{P}(Q)$

	\item
	$V_\Diamond = \mathcal{P}(Q) \cup \{(s,\sigma) ~|~ s \cap
	Q_l \neq \emptyset\}$
	\item
	$V_\Box = \{(s,\sigma) ~|~ s\cap Q_l=\emptyset\}$

	That is, nodes correspond to either sets of states of $A$
	with a letter from $\Sigma$ or just a set of
	states of $A$.
	A node is in $V_\Diamond$ if either it is a plain subset of states of
	$A$ or if it contains local states of $A$.
	A node is in $V_\Box$ if it does not contain local
	states of $A$.
	\item
	A \emph{choice function} for $a\in \Sigma$ is $d:Q_\vee
	\rightarrow Q$ such that for every $q\in Q_\vee$ we have $d(q) \in
	\delta(q,a)$.
	We denote by $D_a$ all the choice functions for
	$a$ and by $D$ all the choice functions for all letters
	$a\in \Sigma$.

	Let $D_A=D \cup Q_\Diamond \cup \Sigma$ be the set of
	labels.

	Intuitively, an edge $e$ from a node (set of states) $s$ to
	set of states $s'$ corresponds to one of three cases.
	\begin{itemize}
	  \item
		Either $e$ corresponds to a set of transitions taken by
		local states of $A$, in which case $e$ is labeled by the
		choice function associating each existential state in $s$
		to the successor chosen for it.
		\item
		Edge $e$ corresponds to a set of transitions taken by
		modal states of $A$.
		In this case the edge corresponds to
		the transitions of exactly \emph{one} existential modal state
		and potentially many universal modal states.
		In this case
		$e$ is labeled by the existential modal state whose
		transition was taken.
		\item
		Or $e$ corresponds to a choice of a letter in $\Sigma$.
	\end{itemize}
	\item
	Given a set of states $s\subseteq Q$, $q\in s\cap
	Q_\Diamond$, a letter $\sigma$, and
	a choice $d\in D_\sigma$ we define the update of $s$ as
	follows:
	$$
	\begin{array}{r c l}
	update_l(s,\sigma,d) &=& (s\setminus Q_l) \cup
	\left \{ q' \left |
	\begin{array}{l}
	\exists q\in s\cap Q_\wedge \mbox{ and } q'\in
	\delta(q,\sigma) \mbox { or} \\
	\exists q\in s \cap Q_\vee \mbox{ and } q'=d(q)
	\end{array}
  \right \} \right . \\
  update_m(s,\sigma,q) &=& \{q' ~|~ q'\in
  \delta(q,\sigma)
  \mbox{ or } \exists q''\in s\cap Q_\Box \mbox { and } q'\in
  \delta(q'',\sigma)\}
	\end{array}
	$$
	That is, a local update consists of the set of all the
	successors of all the
	local universal states in $s$ and all the chosen successors
	of all the local existential states in $s$.
  A modal update consists of the successors of the (modal
  existential) state $q\in s$ and all the successors of all
  the modal universal states in $s$.

	The set of edges is:
	$$
	\begin{array}{r c l r @{,}c @{,} l c l l}
	E & = &
	\bigg \{& ((s,\sigma) & d & (s',\sigma)) & \bigg | &
	\begin{array}{l}
	(s,\sigma) \in V_\Diamond, d\in D_\sigma, \mbox{ and}
	\\
	s'=update_l(s,\sigma,d)
	\end{array}
	 \bigg \} & \cup \\[10pt]
	&& \bigg \{ & ((s,\sigma)&q&s') & \bigg |&
	\begin{array}{l}(s,\sigma)\in V_\Box, q\in
	s\cap Q_\Diamond, \\
	s'= update_m(s,\sigma,q)
	\end{array} \bigg \} & \cup \\[10pt]
	&& \{& (s&\sigma&(s,\sigma)) &|& \sigma\in \Sigma \}
	\end{array}
	$$
	That is, a node $(s,\sigma)$, where $s$ contains local
	states of $A$, has successors that correspond to taking a
	transition from all the local states. For existential local
	states only one successor is taken (according to the choice
	labeling the edge) and for universal local states all
	successors are taken.
	A node $(s,\sigma)$, where $s$ contains no local states, has
	successors that correspond to taking a
	transition from one existential modal state in $s$
	(according to the state labeling the edge) and
	taking the transitions of all the universal modal states in
	$s$.
	A node $s\subseteq Q$, has successors that correspond to choosing a
	letter $\sigma\in \Sigma$ (labeling the edge) and moving to
	$(s,\sigma)$.
\end{itemize}
\end{defn}

We now define the winning condition associated with the
strategy arena.
For a labeled arena, the winning condition is a subset of
$V \cdot (D_A \cdot V)^\omega$.
We construct a word automaton to define the winning condition.

\begin{defn}[Tracking Automaton]
Given an alternating parity tree automaton
$A=(\Sigma,Q,q_0,\delta,\Omega)$,
and the arena $G_A$,
we define the \emph{tracking automaton}
$T_A=(\Sigma_A,Q,q_0,\Gamma,\overline{\Omega})$ to be
a nondeterministic parity word automaton.
The alphabet of $T_A$ is $\Sigma_A=(\Sigma \times (D\cup Q_\Diamond))
\cup \Sigma$.
That is $T_A$ reads either a letter in $\Sigma$ and
either a choice function or a modal existential state or simply a
letter in $\Sigma$.
The transition function of $T_A$ is defined by putting
\begin{align*}
\Gamma(q,(\sigma,d))=\begin{cases}
d(q) & \text{if }q\in Q_\vee \mbox{ and } d\in D\\
\emptyset & \text{if }q\in Q_\vee \mbox{ and } d\in Q_\Diamond\\
\delta(q,\sigma) & \text{if }q\in Q_\wedge\\
\delta(q,\sigma) & \text{if }q \in Q_\Diamond \mbox{ and }
d=q\\
\emptyset & \text{if }q\in Q_\Diamond \mbox{ and }
d\in Q \setminus \{q\}\\
q & \text{if }q \in Q_\Diamond \mbox{ and } d\notin Q\\
\delta(q,\sigma) & \text{if }q\in Q_\Box \mbox{ and } d\in
Q \\
q & \text{if } q\in Q_\Box \mbox{ and } d\notin Q\\
\end{cases}
\end{align*}
for $q\in Q$, $\sigma\in\Sigma$, and $d\in D\cup Q_\Diamond$.

\noindent
For $q\in Q$ and $\sigma\in \Sigma$ we put $\Gamma(q,\sigma) = \{q\}$.
\end{defn}
Notice that the only transitions of $T_A$ that lead to sets
with more than one element are
when $q\in Q_\wedge$.
Indeed, when $q\in Q_m$, by assumption, we have
$|\delta(q,\sigma)|=1$.
Finally, $\overline{\Omega}$ is obtained from $\Omega$ by setting
$\overline{\Omega}(q)=\Omega(q)+1$.

We are now ready to define the acceptance condition $\alpha$.
Recall, that a labeled play $\pi$ in $G_A$ is $\pi\in V\cdot(D_A\cdot
V)^\omega$.
Given a pair $(v,d)$, their projection onto $\Sigma_A$, denoted
$\lfloor v,d \rfloor$, is $\lfloor (s,\sigma),d \rfloor =
(\sigma,d)$ and $\lfloor s,\sigma \rfloor = \sigma$.
Given an infinite sequence $\pi=v_0d_0v_1d_1\cdots\in V\cdot (D_A\cdot
V)^\omega$, we denote by $\lfloor \pi \rfloor$ the sequence $\lfloor
v_0,d_0 \rfloor \lfloor v_1,d_1\rfloor \cdots$.
We define $\alpha_A$ as follows.
$$
\alpha_A = \{ \pi ~|~ \lfloor\pi\rfloor \notin L(T_A)\}
$$

Let $G^+_A$ be the combination of the arena $G_A$ with $\alpha_A$.
We sum up the relation between $A$ and $G_A^+$ as follows.

\begin{theorem}[Simulation]
Let $A$ be an alternating tree automaton $A$.
Then $A$ is non-empty if and only if \pdiam
wins $G^+_A$ from $\{q_0\}$.
\label{theorem:simulation}
\end{theorem}


\begin{toappendix}
\end{toappendix}

\begin{appendixproof}
Let $A=(\Sigma,Q,q_0,\delta,\Omega)$ be non-empty so that there is some Kripke structure
$K=(W,w_0,R,L)$ such that \pdiam wins $(w_0,q_0)$ in $G_{A,K}$. Let
$\kappa_{G_{A,K}}$ be a \emph{positional} strategy for \pdiam in $G_{A,K}$ witnessing this.
We inductively construct a history-dependent strategy $\kappa_{G^+_A}$ for \pdiam in
$G^+_A$ as follows.
Let $\tau=v_0 d_0v_1d_1\ldots v_o\in V\cdot (D\cdot V)^*$
be a finite labeled play of $G^+_A$ according
to the strategy that has been constructed so far.
We write $v_j=s_j$ if $v_j\in \pow(Q)$ and $v_j=(s_j,\sigma_j)$
if $v_j\in \pow(Q)\times\Sigma$.
We will use the following invariant in our construction:
\begin{quote}
	There is a fixed state $w_{o}\in W$ such that
	for all $q\in s_o$, \pdiam wins $(w_{o},q)$ in
	$G_{A,K}$.
\end{quote}
Initially, the invariant holds trivially since \pdiam wins $(w_0,q_0)$ in $G_{A,K}$ by assumption.
\begin{itemize}
	\item
	If $v_o\in \pow(Q)$, then we put
	$\kappa_{G^+_A}(\tau)=(s_o,L(w_o))$ and label
	this transition with $L(w_o)$.
	The invariant trivially holds for $\tau\cdot L(w_o)\cdot
	(s_o,L(w_o))$.
	\item
	If $v_o\in \pow(Q)\times\Sigma$ and
	$s_o\cap Q_l\neq \emptyset$ then we put
	$\kappa_{G^+_A}(\tau)=(s',\sigma_o)$ where
	$s'=update_l(s_o,\sigma_o,d)$
	and where $d$ is obtained by putting $d(q)=q'$ for $q\in s_o\cap
	Q_\vee$,
	where $q'$ is the state such that $\kappa_{G_{A,K}}(w,q)=(w,q')$.
	Label this transition with $d$.
	Since $\kappa_{G_{A,K}}$ is a winning strategy and since
	\pdiam wins $(w_{o},q)$ in $G_{A,K}$ for each $q\in s_o$ by the
	inductive
	invariant, \pdiam also wins $(w_{o},q'')$ in $G_{A,K}$ for each
	$q''\in s'$,
	showing that the invariant holds for $\tau\cdot d\cdot (s',\sigma_o)$.
	\item
	If $v_o\in \pow(Q)\times\Sigma$ and $s_o\cap Q_l= \emptyset$, then we
	have to show that the invariant holds for all successors of $v_o$.

	Consider $s'$ and $q\in s_o\cap Q_\Diamond$ such that we have
	$((s_o,\sigma_o),q,s')\in E$.
	Recall that \pdiam wins $(w_o,q'')$ for all $q''\in s_o$ by
	the inductive invariant.
	In particular, \pdiam wins $(q,w_o)$ in $G_{A,K}$ so that there is
	$w_{o+1}\in R(w)$ such that \pdiam wins
	$(\delta(q,\sigma_o),w_{o+1})$ in $G_{A,K}$.
	For all $q''\in s_o\cap
	Q_\Box$ and $q'\in \delta(q'',\sigma_o)$,
	\pdiam wins $(q',w')$ for all $w'\in R(w)$ since
	\pdiam wins $(q'',w_o)$ in $G_{A,K}$; in particular, \pdiam wins
	$(q',w_{o+1})$.
	Hence the invariant holds for $\tau\cdot q\cdot s'$.
\end{itemize}
The function $\kappa_{G^+_A}$ is a strategy for \pdiam in $G^+_A$ by construction.
To see that $\kappa_{G^+_A}$ is a winning strategy, let $\pi=v_0 d_0 v_1\ldots$ be
a labeled play of $G^+_A$ that follows $\kappa_{G^+_A}$.
It remains to show that $\lfloor\pi\rfloor\notin L(T_A)$.
In the case that $\pi$ is finite, clearly $\lfloor\pi\rfloor\notin
L(T_A)$.
Consider the case that $\pi$ is infinite.
In this case, $\pi$ corresponds to an infinite play in $G_{A,K}$ with
potentially some finite stuttering.
Let $\pi=v_0d_0\ldots$ and let $w_0w_1\ldots$ be the worlds of $K$ in
the construction of the invariant above.
Any run of $T_A$ on $\lfloor\pi\rfloor$ is a sequence $q_0\ldots$ maintaining
$q_o\in s_o$ forall $o\geq 0$, where $v_o$ is either of the form
$(s_o,\sigma_o)$ or of the form $s_o$.
Furthermore, transitions where $T_A$ does not change its state due to
our construction stutter a finite number of times.
It follows that $q_0q_1\ldots$ is winning for \pdiam in $G_{A,K}$ and
thus cannot satisfy the parity condition of $T_A$, which is
$\overline{\Omega}$.

For the converse direction, let \pdiam win $G^+_A$ from $\{q_0\}$.
Let $\kappa_{G^+_A}$ be a strategy for \pdiam in $G^+_A$ that witnesses this.
We inductively construct a Kripke structure $K=(W,w_0,R,L)$ and a
strategy $\kappa_{G_{A,K}}$.
To this end, we proceed by induction over the length of prefixes
$(w_0,q_0)\cdots (w_o,q_o)$ of plays in $G_{A,K}$.
Let $o$ be a number such that $K$ and $\kappa_{G_{A,K}}$ have been
constructed for all plays $(w_0,q_0)\cdots (w_o,q_o)$ of $G_{A,K}$.
We use the following inductive invariant.
\begin{quote}
for all plays $\pi=(w_0,q_0)\cdots (w_o,q_o)$ of
$G_{A,K}$, there exists a play $v_0,\ldots v_l$ of $G^+_A$ compatible
with $\kappa_{G^+_A}$ and $v_l=(s_l,\sigma_l,i_l)$ where
$L(q_o)=\sigma_l$ and $q_o\in s_l$.
\end{quote}

Initially, $G^+_A$ starts from $\{q_0\}$ and $\kappa_{G^+_A}$ chooses a
successor $(\{q_0\},\sigma,1)$ for some $\sigma\in \Sigma$.
We add to $K$ the initial world $w_0$ such that $L(w_0)=\sigma$ and
consider the prefix $(w_0,q_0)$ in $G_{A,K}$.
We have $(w_0,q_0)$ associated with the prefix
$\{q_0\},(\{q_0\},\sigma,1)$ compatible with $\kappa_{G^+_A}$.

In the inductive step, we extend $K$
and $\kappa_{G_{A,K}}$ as follows.
Consider a prefix $(w_0,q_0)\cdots (w_o,q_o)$ and the associated play
$v_0,\ldots, v_l$ of $G^+_A$. Let $v_l=(s_l,\sigma_l)$.

We start by extending $K$.
We build temporarily a function $f_{w_o}$ that associates a
newly created successor world $w_q$ of $w_o$ with combinations of
states $q\in Q_\Diamond$ and plays
$v_0,\ldots,v_l,v_{l+1},\ldots, v_{l'}$ that are compatible with $\kappa_{G^+_A}$.
By construction either $v_l\in V_\Box$ or $\kappa_{G^+_A}$ goes
through a finite
sequence $v_{l+1},\ldots,v_{l+p}$ of nodes in $V_\Diamond$ until it
gets to a node $v_{l+p+1}$ in $V_\Box$.
This sequence is finite due to the automaton having no loops of local
states.
For each state $q\in s_{l+p+1}\cap Q_\Diamond$, there exists a
successor $v_q$ such that $(v_{l+p+1},q,v_q)$ is an edge in $G^+_A$.
It follows that $v_q$ is of the form $s_q$ and that
$q\in s_q$.
Furthermore, $\kappa_{G^+_A}$ chooses a successor
$v'_q=(s_q,\sigma_q)$ of $v_q$.
We extend $K$ by adding to it a new world $w_q$ that is a
successor of $w_o$ such that $L(w_{q})=\sigma_q$.
The function $f_{w_o}$ associates $w_q$ with $q$ and the play
$v_0,\ldots, v_l,v_{l+1},\ldots, v_{l+p},v_{l+p+1},v_q$.

In order to extend $\kappa_{G_{A,K}}$
we consider the following cases:
\begin{itemize}
	\item
	$q_o \in Q_\vee$ - by assumption $q_o\in s_l$.
	As $q_0\in s_l$ we know that $v_l$ is in $V_\Diamond$.
	Hence, $\kappa_{G^+_A}$ singles out a successor $v_{l+1}$ of $v_l$.
	Let $d$ be the label of the edge $(v_{l},d,v_{l+1})$.
	Let $q_{o+1}=d(q_o)$.
	By definition $q_{o+1}\in \delta(q_o,\sigma_l)=\delta(q_o,L(w_o))$.
	Furthermore, we maintain that $q_{o+1}\in s_{l+1}$.
	We extend the strategy $\kappa_{G_{A,K}}$ to choose $(q_{o+1},w_o)$
	as the successor of $(q_o,w_o)$.
	We associate this extended play in $G_{A,K}$ with the play
	$v_0,\ldots, v_{l+1}$ in $G^+_A$.
	\item
	$q_o \in Q_\wedge$ - by assumption $q_o\in s_l$.
	We continue as above except that for every $q'\in \delta(q_o,L(w_o))$
	we add a successor $(q',w_o)$ as compatible with $\kappa_{G_{A,K}}$.
	The successor $(q',w_o)$ is associated with the same extended play in
	$G^+_A$ as above.
	\item
	$q_o \in Q_\Diamond$ - by assumption $q_o\in s_l$.
	Let $q_{o+1}$ be the unique state such that $q_{o+1}\in
	\delta(q_o,L(w_o))$.
	\begin{itemize}
		\item
		If $v_l\in V_\Box$ then $v_l$ has been used in the construction of
		$K$ to identify the successor $w_{q_o}$ of $w_o$ that is associated
		with $q_o$ in the construction of $K$.
	  It follows that $q_{o+1}\in s_{q_o}$.
	  We extend the play compatible with $\kappa_{G_{A,K}}$ by adding to
	  it the pair $(w_{q_o},q_{o+1})$ and associate with it the play
	  pre-prepared for $w_{q_o}$ in the construction of $K$.
		\item
		If $v_l\in V_\Diamond$ then $v_l$ has a unique
		descendant compatible with $\kappa_{G^+_A}$ that has been used in
		the construction of $K$ and we proceed as above.
	\end{itemize}
	\item
	$q_o \in Q_\Box$ - as before, let $q_{o+1}$ be the unique state such
	that $q_{o+1}\in \delta(q_o,L(w_o))$.
	We now proceed similar to the case of $q_o \in Q_\Diamond$.

	As before, we find that $q_o$
	is in the node $v'$ used to construct all the successors of $w_o$.
	For every successor $w_q$ of $w_o$ we create a play extended by
	$(w_q,q_{o+1})$ and associate it with the pre-prepared play that ends
	in $w_q$.
	The construction of the arena $G^+_A$ ensures that $q_{o+1}\in s_q$.
\end{itemize}

Let $K^\infty$ and $\kappa_{G_{A,K^\infty}}$ denote the Kripke structure
and strategy, respectively, that are obtained in this way.
It remains to show that $\kappa_{G_{A,K^\infty}}$ is a winning strategy
for \pdiam in $G_{A,K^\infty}$.
Consider an infinite play $\pi=(w_0,q_0),\ldots$ that is an outcome of
$\kappa_{G_{A,K^\infty}}$.
During the construction of $K^\infty$ and $\kappa_{G_{A,K^\infty}}$
each prefix $(w_0,q_0),\ldots, (w_o,q_o)$ of $\pi$ was associated with
a prefix $v_0,\ldots,v_l$ of a play in $G^+_A$ compatible with
$\kappa_{G^+_A}$.
It follows that the limit of all these prefixes is an infinite play
compatible with $\kappa_{G^+_A}$ and hence winning in $G^+_A$.
Let $d_0,\ldots$ be the sequence of labels such that
$(v_i,d_i,v_{i+1})$ are the edges in $G^+_A$ taken in this play.
It follows that the sequence of state $q_0,\ldots$ corresponds to a run
of $T_A$ reading $\lfloor v_0,d_0\rfloor \lfloor v_1,d_1\rfloor \cdots$
with (potentially) additional stuttering resulting from entering a
modal state $q_i$ while the matching $v_{i'}$ is of the form
$(s,L(w_i))$ where $s\cap Q_l\neq\emptyset$.
As the stuttering is finite, it follows that $\inf(q_0,\ldots)$ is the
same as for the run of $T_A$.
As the run of $T_A$ is rejecting, the sequence $q_0,\ldots$ is winning
for \pdiam.
\hfill \qed
\end{appendixproof}

Consider the automaton $T_A$.
By Theorem~\ref{theorem:history determinization and determinization},
there exists an equivalent history deterministic automaton $H_A$.
Let $H_A=(\Sigma_A,T,t_0,\rho,\Omega')$.
By using $H_A$ we can turn the game $G^+_A$ to a parity game $G^*_A$
capturing the non-emptiness of $A$.

\begin{defn}[$G^*_A$]
	Consider the strategy arena $G_A=(V,V_\Diamond,V_\Box,E)$
	and the automaton $H_A$.
	We construct the parity game $G^*_A =
	(V',V'_\Diamond,V'_\Box,E',\Omega'')$, where the components of
	$G^*_A$
	are as follows.
	\begin{itemize}
		\item
		$V'=(V \times T) \cup (\Sigma_A \times V \times T)$
		\item
		$V'_\Diamond = V_\Diamond \times T$
		\item
		$V'_\Box = (V_\Box \times T) \cup (\Sigma_A \times V \times T)$
		\item
		The set of edges is:
		$$
		\begin{array}{r c @{\{} r @{~,~} l @{~|~} l @{\}\quad}l}
			E' & = &
			(((s,\sigma),t) & ((\sigma,d),v',t) & ((s,\sigma),d,v')
			\in E &
			\cup \\
			&& ((s,t) & (\sigma,v',t)) & (s,\sigma,v')\in E & \cup \\
			&& (((\sigma,d),v',t) & (v',t')) & t'\in
			\rho(t,(\sigma,d)) & \cup \\
			&& ((\sigma,v',t) & (v',t')) & t' \in \rho(t,\sigma)
		\end{array}
		$$
		\item
		The priority function $\Omega''$ is obtained from $\Omega'$ by
		setting $\Omega''(v,t) = \Omega'(t)+1$ and $\Omega''(a,v,t) =
		\Omega'(t)+1$.
	\end{itemize}
\end{defn}

\begin{theorem}[Game Translation]
	Player~$\Diamond$ wins in $G^+_A$ from some state
	$(s,\sigma)$ if and only if \pdiam wins in $G^*_A$ from
	$((s,\sigma),t_0)$.
\label{theorem:game translation}
\end{theorem}

\begin{proof}
	We can show that \pbox wins in $G^*_A$ if and only if she wins in $G^+_A$.
	The proof follows the proof that history deterministic automata can
	be used in combination with games as in \cite{HenzingerP06}.
\end{proof}

\begin{remark}
An alternative way to view the construction of $T_A$ and $H_A$ is to
think about the dual of $T_A$ as a universal automaton recognizing
plays that are winning for \pdiam.
Then, the dual of $H_A$ would be a history-deterministic universal
parity automaton recognizing the same language.
The resolution of the transition function of a history-deterministic
universal automaton is delegated to \pbox just like it is in the
construction of $G^*_A$.
In particular, every history-determinization for nondeterministic
automata is, in fact, also a history-determinization for universal
automata.
This implies that the history-determinization construction of Henzinger
and Piterman \cite{HenzingerP06} can be also used for
under-approximating the losing region in an LTL game, which was left as
an open question in their paper.
\end{remark}

The following is a direct implication of
Theorems~\ref{theorem:simulation} and
\ref{theorem:game translation}.

\begin{corollary}[Emptiness]
Let $A$ be an alternating tree automaton.
Then $A$ is non-empty if and only if \pdiam
wins $G^*_A$ from the node $(\{q_0\},t_0)$.
\end{corollary}

We now consider the complexity of the decision problem.
Let $A$ be an alternating tree automaton reading an alphabet of
size $m$ with $n$ states and rank $k$.
Then $G_A$ has $(m+1)\cdot 2^n$ vertices and $T_A$ has $n$ states and
rank $k$ as well.
Let $S_{HD}(n,k)$ denote the number of states and $R_{HD}(n,k)$ denote
the rank of a history deterministic automaton obtained from a
nondeterministic word automation with $n$ states and rank $k$.

\begin{corollary}[Complexity]\label{cor:treeautempt}
	Let $A$ be an alternating tree automaton reading an alphabet of size
	$m$, with $n$ states and rank $k$.
	The complexity of emptiness of $A$ is
	\paritytime{(m+1)\cdot 2^n \cdot S_{HD}(n,k)}{R_{HD}(n,k)}.
\end{corollary}

\begin{remark}
In case that $H_A$ is a deterministic automaton,
a winning strategy for \pdiam in $G^*_A$ directly induces a Kripke structure $K$ with set
of worlds $W=V_\Box\times T$ such that $A$ accepts $K$. Hence
non-empty alternating parity tree automata accept some structure of
size at most $(m+1)\cdot 2^n\cdot S_{HD}(n,k)$ which is
in $\mathcal{O}((m+1)\cdot 2^n\cdot ((nk)!)^2)$ by~\autoref{lem:parity2buchi}
and~\autoref{lem:sp} below.
\end{remark}

\section{Transformations of Word Automata}
\label{section:word automata}

We specialize parity acceptance conditions to the special cases of
\buchi and \cbuchi conditions.
In a \buchi condition the priority function uses only the priorities
$1,2$.
In a \cbuchi condition the priority function uses only the priorities
$0,1$.
For \buchi automata, we put $F=\{q\in Q\mid \Omega(q)=2\}$;
for \cbuchi automata, we put $F=\{q\in Q\mid \Omega(q)=0\}$;
in both cases, we put $\overline{F}=Q\setminus F$.
The \buchi acceptance requires
accepting runs to contain infinitely many accepting states, while \cbuchi acceptance requires
accepting runs to contain only finitely many non-accepting states.
A \cbuchi automaton is weak if for all its strongly connected components $C$,
we have $C\subseteq F$ or $C\subseteq \overline{F}$.
For deterministic automata, we extend $\delta$ from letters to finite
words in the obvious way.

\begin{lemma}[\hspace{-0.2pt}\cite{KKV01}]
	Let $A=(\Sigma,Q,q_0,\delta,\Omega)$ be a nondeterministic parity
	word automaton
	 of rank $k$.
	Then there is a nondeterministic \buchi word automaton
	$A'=(\Sigma,Q',q_0,\delta',F)$ such that $L(A)=L(A')$ and $|Q'|\leq
	\left ( \left \lceil \frac{k+1}{2} \right \rceil +1 \right)\cdot |Q|$.
	\label{lem:parity2buchi}
\end{lemma}

\begin{proof}
    We just recall the construction of $A'$ and refer to
    \cite{KKV01,HausmannEA18} for the proof of $L(A)=L(A')$.
    Intuitively, the automaton $A'$ nondeterministically guesses a
    position and an even priority $p$ such that there is a run of $A$
    on the input word such that from the guessed position on,
    no state with priority greater than $p$ is visited and some state
    with priority $p$ is visited infinitely often. Formally,
	we put $Q''=Q\times \{i\in\mathbb{N}\mid 0\leq i\leq  k\text{ and }i
	\text{ is even}\}$ and
	\begin{align*}
		Q'&=Q\cup Q''&
		F&=\{(q,i)\in Q''\mid\Omega(q)=i\}
	\end{align*}
	so that the claimed bound on the size of $A'$ follows immediately.
	The transition function $\delta'$ is defined, for $q\in Q$, even $i$
	such that
	$0\leq i\leq k$ and $a\in \Sigma$,
	by putting
	\begin{align*}
		\delta'(q,a)=\delta(q,a)\cup \{(q',i)\in Q''\mid
		q'\in\delta(q,a)\text{ and }\Omega(q')=i\}
	\end{align*}
	and
	\begin{align*}
		\delta'((q,i),a)=\{(q',i)\in\delta(q,a)\times\{i\}\mid
		\Omega(q')\leq i\}.
	\end{align*}
	\qed
\end{proof}



\begin{defn}[Limit-linear \cbuchi Automata]
A \cbuchi automaton $A=(\Sigma,Q,q_0,\delta,F)$
is \emph{limit-linear} if
for all $q\in F$, there is exactly one $\delta$-path that stays
in $F$ and leads from $q$ to $q$.
\end{defn}

\begin{defn}[Limit-deterministic Word Automata]
Fix a parity word automaton $A=(\Sigma,Q,q_0,\delta,\Omega)$.
Given a state $q\in Q$, the \emph{compartment} $C_q$ of $q$ consists of
all states that are reachable from $q$
by a path that visits states with priority at most $\Omega(q)$.
We say that $A$ is \emph{limit-deterministic (LD)}
if for all states $q$ such that $\Omega(q)$ is even, $C_q$ is \emph{internally deterministic},
that is, $|\delta(q',a)\cap C_q|\leq 1$ for all $q'\in C_q$ and $a\in \Sigma$.
\end{defn}
Thus a \buchi automaton is limit-deterministic if
all its states that are reachable from an accepting state are deterministic.
A \cbuchi automaton is limit-deterministic if all its accepting states are deterministic.
In particular, every limit-linear \cbuchi automaton is limit-deterministic.

\begin{lemma}
	The construction in Lemma~\ref{lem:parity2buchi} preserves limit
	determinism.
\end{lemma}

\begin{proof}
	Let $A$ be a limit-deterministic parity automaton.
	Then we claim that $A'$ as constructed in
	Lemma~\ref{lem:parity2buchi} is limit-deterministic.
	Since $A'$ is a \buchi automaton and since all states that are
	reachable from
	some state in $F$ are contained in $Q\times\{p\}$ for some even $p$,
	it suffices to show that for all even $p$, all
	states $(q,p)\in Q\times\{p\}$ and all $a\in \Sigma$, we have
	$|\delta'((q,p),a)|\leq 1$.
	So let $(q,p)\in Q\times\{p\}$.
	Then, by construction of $A'$, $(q,p)$ is contained in the
	compartment $C$ of some
	state with priority $p$ in $A$.
	By definition of $\delta'$, we have $\delta'((q,p),a)\subseteq C$
	since $C$ is a compartment. Since $A$ is limit-deterministic, $C$ is
	internally
	deterministic which shows $|\delta'((q,p),a)|\leq 1$, as required.
	\qed
\end{proof}

\subsection{Determinizing Word Automata}

We give specialized determinization constructions for limit-linear
\cbuchi automata, nondeterministic \cbuchi automata,
limit-deterministic \buchi automata, and finally for general \buchi
automata and parity automata.


\begin{lemma}[Circle method]\label{lem:limitlinear}
	Let $A$ be a limit-linear \cbuchi automaton with $n$ states.
	Then there is a deterministic \cbuchi automaton $A'$
	with $n'$ states such that $L(A)=L(A')$ and $n'\leq n^2\cdot 2^{n}$.
\end{lemma}
\begin{proof}
    Let $A=(\Sigma,Q,q_0,\delta,F)$. If
    $F=Q$, then $A$ is deterministic and we put $A'=A$.
    Otherwise, we have $|F|<n$ and proceed with the following construction,
    which is similar to the powerset construction, but additionally
    annotates macro-states with a single state and a counter.
    The states of accepting components are arranged in a cycle since
	$A$ is limit-linear.
    Intuitively, the single state component of macro-states identifies exactly one state in
    exactly one accepting cycle that has a token. The determinized automaton then checks whether
    it is possible to stay within this cycle forever, moving
    the token according to the letters that are read. If this is not possible, the automaton
    reduces the counter by one and moves the token to \emph{the} next state in the current
    cycle and again checks whether is possible to stay in the cycle forever
    when moving the token according to the read word.
    When this fails so often that the counter reaches $0$, the automaton picks a state
    from another accepting cycle, moves the token to this state
    and resets the counter. It is crucial that the moving of tokens between
    accepting cycles is done in a fair way, so that
    if the token changes cycles infinitely often, the
    token visits every accepting cycle infinitely often.
    Then the token eventually stays forever within one accepting cycle
    if and only if there is an accepting run.

	Formally, we proceed as follows. For moving the token between accepting
	cycles, we assume a function $\mathsf{next}:Q\to Q$ such that
	for $q\in Q$, $\mathsf{next}(q)$ is some abitrary but fixed state
	from an	accepting cycle of $A$ such that
	iterative application of $\mathsf{next}$ cycles through all
	accepting cycles of $A$ in a fair manner.
    We also assume a function $\mathsf{step}:F\to F$ that cycles
	through the states of a single accepting cycle of $A$ in a fair
	manner;
	formally, we put $\mathsf{step}(q)=q'$ where $q'$ is \emph{the} state
	such that
	there is some $a\in\Sigma$ such that $\delta(q,a)\cap F = \{q'\}$.
	We define the deterministic \cbuchi automaton $A'=(\Sigma,Q',u_0,\delta',F')$
	by putting
	\begin{align*}
		Q'&=2^Q\times Q\times \{0,\ldots, |F|\} & F'&=\{(U,q,c)\in Q'\mid c\neq 0\}
	\end{align*}
	and $u_0=(\{q_0\},q_0,0)$. The claimed bound on the size of $A'$ follows
	immediately since $|F|<n$ so that $|\{0,\ldots,|F|\}|\leq n$.
	Finally, the transition relation $\delta'$ is defined by putting,
    for $(U,q,c)\in Q'$ and $a\in\Sigma$,
	\begin{align*}
		\delta'((U,q,c),a)=(\delta(U,a),q',c)
	\end{align*}
	if $c\neq 0$, $\delta(q,a)\cap F=\{q'\}$ and $q\in U$; this
	moves the token within the current accepting cycle according to the
	input letter, if possible. Otherwise, the run represented by the token
	does not stay in the current accepting cycle and we move the token to another
	state. This is achieved by putting
	\begin{align*}
		\delta'((U,q,c),a)=\begin{cases}
     		(\delta(U,a),\mathsf{next}(q),|F|) & \text{ if } c=0\\
			(\delta(U,a),\mathsf{step}(\mathsf{step}(q)),c-1) & \text{ if } c>0
		\end{cases}
	\end{align*}
	If $c=0$, then the token is moved to the next accepting cycle and the
	counter is reset to $|F|$; if
	$c>0$, then the token is moved to the next state in the current accepting cycle
	(to also incorporate the $a$-transition that takes place, we apply
	$\mathsf{step}$ twice) and the counter is reduced by $1$.

	\begin{toappendix}

		To see $L(A)\subseteq L(A')$, let $\tau=q_0,q_1,\ldots$ be an
		accepting
		run of $A$ on some word $w$. Then $\tau$ stays within one accepting
		component $C$
		of $A$ from some point on; let $i$ be a position from which on
		this the case. Since $A$ is limit-linear, $C$ forms a cycle.
		Let $\pi$ be the run of $A'$ on $w=a_0,a_1,\ldots$. We have to show
		that there is a position $j$
		such that all states in $\pi$ are of the shape
		$(U_{j'},q'_{j'},c_{j'})$ from position $j$ on.
		If there is some position $j$ and some accepting component $C'\neq
		C$ such that
		the states in $\pi$ are of the shape $(U_{j'},q'_{j'},c_{j'})$ such
		that
		$q'_{j'}\in C'$ for all $j'\geq j$, then we are done. Otherwise, it
		suffices
		to show that there is a position $j\geq i$ such that the $j-th$
		state in $\pi$ is of the shape $(U_j,q_j,c_j)$, since
		we then have $\delta(q_{j'},a_{j'})\cap F=\{q_{j'+1}\}$ and
		$q_{j'}\in U_{j'}$
		for all $j'\geq j$, since $A$ is limit linear.
		Since $\mathsf{next}$ cycles through the strongly connected
		components
		of $A$ in a fair manner and since there is no position $j$ such that
		there is an accepting component $C'\neq C$ such that
		the states in $\pi$ are of the shape $(U_{j'},q'_{j'},c_{j'})$ such
		that
		$q'_{j'}\in C'$ for all $j'\geq j$, there is some $j''$ such that
		the $j''$-th state in $\pi$ is of the shape
		$(U_{j''},q'_{j''},c_{j''})$, where
		$q'_{j''}\in C$. We proceed by induction of the length
		$m(q'_{j''},q_{j''})$ of
		\emph{the} path from $q'_{j''}$ to $q_{j''}$. If
		$m(q'_{j''},q_{j''})=0$,
		then $q'_{j''}=q_{j''}$ so that we are done. If
		$m(q'_{j''},q_{j''})>1$ then
		we distinguish cases. If $\delta(q'_{j''},a_{j''})\cap
		F\neq\emptyset$ and
		$q'_{j''}\in U_{j''}$ for all $j''\geq j'$, then we are done.
		Otherwise, we eventually reach a state $(U_o,q_o,c_o)$ such that
		$\delta(q'_{o},a_{o})\cap F=\emptyset$ or $q'_{o}\in U_{o}$.
		For the next state $(U_{o+1},q_{o+1},c_{o+1})$ we then have
		$c_{o+1}=c_o-1$ and
		$q'_{o+1}=\delta(\mathsf{step}(q'_o),a)$. We have
		$m(q'_{o+1},q_{o+1})<m(q'_{o},q_{o})$
		so that the inductive hypothesis finishes the case.

		For the converse direction, we have to show $L(A')\subseteq L(A)$.
		Let $w=a_0,a_1,\ldots$ be a word such that $w\in L(A')$ and let
		$\pi$ be the accepting run of $A'$ on $w$. Since $\pi$ is
		accepting, there is a position
		$i$ such that the
		states in $\pi$ are, for all $i'\geq i$, of the shape
		$(U_{i'},q_{i'},c)$ for some $c$
		and we have $\delta(q_{i'},a_{i'})\cap F=\{q_{i'+1}\}$ and
		$q_{i'}\in U_{i'}$;
		in particular, $q_{i'}\in F$ for all $i'>i$. Let $\tau_1$ be a run
		of $A$ on
		$a_0,\ldots,a_i$ that ends in $q_i$ and put
		$\tau=\tau_1;q_{i+1},q_{i+2},\ldots$
		so that $\tau$ is an accepting run of $A$ on $w$.
	\end{toappendix}
	\qed
\end{proof}

\begin{example}
Consider the limit-linear \cbuchi automaton $A$
depicted below,
and the equivalent deterministic \cbuchi automaton $A'$ obtained by
using the construction from Lemma~\ref{lem:limitlinear};
to be able show a complete example, $A$ is picked to be a very simple automaton
(accepting just the word $(ab)^\omega$).
For brevity, we depict only the reachable part of $A'$ and collapse
all macro-states of the shape $(\emptyset,q,c)$ to a single non-accepting
sink state $\bot$. Any macro-state
in $A'$ that has a nonzero counter value is accepting.
We have $\mathsf{step}(y)=u$ and $\mathsf{step}(u)=y$.
Since there is just one accepting strongly connected component in $A$,
we assume that $\mathsf{next}(x)=\mathsf{next}(y)=\mathsf{next}(z)=y$,
and $\mathsf{next}(u)=u$.
\medskip

\hspace{10pt}
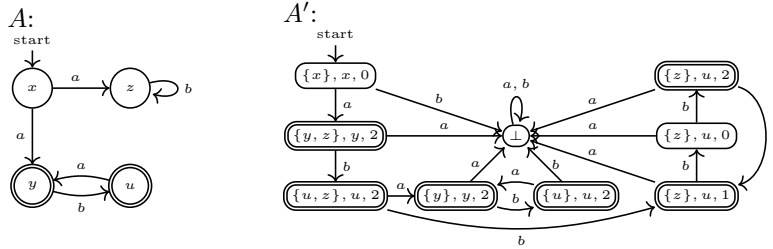
\begin{minipage}{.3\linewidth}
\vspace{-10pt}
		 $A$:\\
\tikzset{every state/.style={minimum size=15pt}}
\begin{tiny}
  \begin{tikzpicture}[
		auto,
    node distance=1.3cm,
    semithick
    ]
     \node[state,initial above] (0) {$x$};
     \node[state,accepting] (1) [below of=0] {$y$};
     \node[state,accepting] (2) [right of=1] {$u$};
     \node[state] (3) [right of=0] {$z$};
     \path[->] (0) edge node [pos=0.5,left] {$a$} (1);
     \path[->] (0) edge node [pos=0.4,above] {$a$} (3);
     \path[->] (1) edge [bend right=15] node [pos=0.5,below] {$b$} (2);
     \path[->] (2) edge [bend right=15] node [pos=0.5,above] {$a$} (1);
     \path[->] (3) edge [loop right] node [right] {$b$} (3);
       \end{tikzpicture}
\end{tiny}

    \end{minipage}%
    \begin{minipage}{.7\linewidth}
$A'$:
\tikzset{every state/.style={minimum size=25pt}}
\begin{tiny}

  \begin{tikzpicture}[
		auto,
    node distance=0.8cm,
    semithick
    ]
     \node[rounded corners, draw,initial above] (1) {$\{x\},x,0$};
     \node[rounded corners, draw,accepting] (2) [below of=1] {$\{y,z\},y,2$};
     \node[rounded corners, draw,accepting] (3) [below of=2] {$\{u,z\},u,2$};
     \node (yo) [right of=3] {};
     \node[rounded corners, draw,accepting] (4) [right of=yo] {$\{y\},y,2$};
     \node (yo1) [right of=4] {};
     \node[rounded corners, draw,accepting] (5) [right of=yo1] {$\{u\},u,2$};
     \node (yo2) [right of=5] {};
     \node[rounded corners, draw,accepting] (6) [right of=yo2] {$\{z\},u,1$};
     \node[rounded corners, draw] (7) [above of=6] {$\{z\},u,0$};
     \node[rounded corners, draw,accepting] (8) [above of=7] {$\{z\},u,2$};
     \node[rounded corners, draw] (9) [above of=yo1] {$\bot$};
     \path[->] (1) edge node [above] {$b$} (9);
     \path[->] (1) edge node [right] {$a$} (2);
     \path[->] (2) edge node [above] {$a$} (9);
     \path[->] (2) edge node [right] {$b$} (3);
     \path[->] (3) edge node [above] {$a$} (4);
     \path[->] (3) edge [bend right=15] node [pos=0.5,below] {$b$} (6);
     \path[->] (4) edge node [left] {$a$} (9);
     \path[->] (4) edge [bend right=15] node [above] {$b$} (5);
     \path[->] (5) edge node [right] {$b$} (9);
     \path[->] (5) edge [bend right=15] node [above] {$a$} (4);
     \path[->] (6) edge node [above] {$a$} (9);
     \path[->] (6) edge node [left] {$b$} (7);
     \path[->] (7) edge node [above] {$a$} (9);
     \path[->] (7) edge node [left] {$b$} (8);
     \path[->] (8) edge node [above] {$a$} (9);
     \path[->] (8) edge [bend left=75] node [pos=0.5,right] {$b$} (6);
     \path[->] (9) edge [loop above] node [above] {$a,b$} (9);

     \end{tikzpicture}
\end{tiny}
\end{minipage}
The automaton $A'$ starts with the token at $x$ and with counter value $0$.
When reading $a$, the token is moved to $\mathsf{next}(x)=y$ and the counter is reset
to $2$. Afterwards, there are two cases: If the automaton
reads $bb$, it is not possible in $A$ to move the token accordingly from
$y$ and stay in the accepting cycle between $y$ and $u$. Thus $A'$ transitions
to $(\{z\},u,1)$, intuitively moving the token to the next accepting cycle,
which in this example moves the token to $u$. This state
however is not contained in the powerset component $\{z\}$ so that the automaton rejects
the word, which is reflected by the fact that $(\{z\},u,1)$ accepts the empty language.
The other option to proceed from $(\{y,z\},y,2)$ is by reading sequences $(ba)^*$,
which results in repeatedly moving the token from $y$ to $u$ and back to $y$;
if this continues forever, the word is accepted by $A'$.
Otherwise, a sequence $aa$ or $bb$ is read eventually and the automaton transitions
to the sink state and rejects the word.

\end{example}

\begin{lemma}[Miyano-Hayashi~\cite{MiyanoHayashi1984}]\label{lem:miyanohayashi}
Given a nondeterministic \cbuchi automaton $A=(\Sigma,Q,q_0,\delta,F)$,
there is a deterministic \cbuchi automaton $A'=(\Sigma,Q',u_0,\delta',F')$ such
that $L(A)=L(A')$ and $|Q'|\leq 3^{|Q|}$.
\end{lemma}

\begin{proof}
We just show the construction of $A'$; for the proof of $L(A)=L(A')$
we refer to~\cite{MiyanoHayashi1984}. The construction
is similar to the powerset construction but additionaly tracks subsets $V$ of the
accepting states $U\cap F$ of macro-states $U\subseteq Q$. Intuitively, there is,
for each state in $V$, a run of $A$ that has not left $F$ recently. Whenever this set is
the empty set, it is reset to all accepting states of the current macro-state.
A run of $A'$ then is accepting if such resetting steps happen only finitely often,
ensuring the existence of a run of $A$ that from some point on stays within $F$ forever.
Formally, we put
\begin{align*}
Q'&=\{(U,V)\mid U\subseteq Q, V\subseteq U\cap F\}&
F'&=\{(U,V)\in Q'\mid V\neq \emptyset\}
\end{align*}
and $u_0=(\{q_0\},\emptyset)$.
The claimed bound on the size of $A'$ follows since macro-states $(U,V)\in Q'$
can be coded by functions $f:Q\to\{0,1,2\}$ where $f(q)=0$ if $q\notin U$,
$f(q)=1$ if $q\in U$ but $q\notin V$ and $f(q)=2$ if $q\in V$; the number of
such functions is bounded by $3^{|Q|}$.
We define $\delta'$ by putting
\begin{align*}
\delta'((U,V),a)=\begin{cases}
(\delta(U,a),\delta(V,a)\cap F) & \text{if }V\neq\emptyset\\
(\delta(U,a),\delta(U,a)\cap F) & \text{if }V=\emptyset
\end{cases}
\end{align*}
for $(U,V)\in Q'$ and $a\in \Sigma$.
%
\qed
\end{proof}

\begin{example}
Consider the nondeterministic \cbuchi automaton $A$
depicted below,
and the equivalent deterministic automaton $A'$ obtained by
using the construction from Lemma~\ref{lem:miyanohayashi};
both automata accept exactly the infinite words over $\Sigma=\{a,b\}$
that contain $a$ finitely often.
For brevity, we depict only the reachable part of $A'$ and
label macro-states $(U,V)$ with $U\setminus V,V$.
\medskip

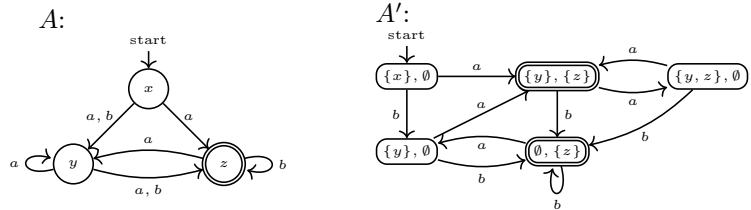
\begin{minipage}{.3\linewidth}
$\quad$
		 $A$:\\
\tikzset{every state/.style={minimum size=15pt}}
\begin{tiny}
  \begin{tikzpicture}[
		auto,
    node distance=1.0cm,
    semithick
    ]
     \node[state,initial above] (0) {$x$};
     \node (yo) [below of=0] {};
     \node[state] (1) [left of=yo] {$y$};
     \node[state,accepting] (2) [right of=yo] {$z$};
     \path[->] (0) edge node [pos=0.3,left] {$a,b$} (1);
     \path[->] (0) edge node [pos=0.3,right] {$a$} (2);
     \path[->] (1) edge [loop left] node [left] {$a$} (1);
     \path[->] (1) edge [bend right=15] node [pos=0.5,below] {$a,b$} (2);
     \path[->] (2) edge [bend right=15] node [pos=0.5,above] {$a$} (1);
     \path[->] (2) edge [loop right] node [right] {$b$} (2);
       \end{tikzpicture}
\end{tiny}

    \end{minipage}%
		$\quad$
		$\quad$
		$\quad$
    \begin{minipage}{.7\linewidth}
$A'$:
\tikzset{every state/.style={minimum size=15pt}}
\begin{tiny}

  \begin{tikzpicture}[
		auto,
    node distance=1.0cm,
    semithick
    ]
     \node[rounded corners, draw,initial above] (1) {$\{x\},\emptyset$};
		 \node (yo) [right of=1] {};
     \node[rounded corners, draw] (2) [below of=1] {$\{y\},\emptyset$};
     \node[rounded corners, draw,accepting] (3) [right of=yo] {$\{y\},\{z\}$};
     \node[rounded corners, draw,accepting] (4) [below of=3] {$\emptyset,\{z\}$};
		 \node (yo2) [right of=3] {};
     \node[rounded corners, draw] (5) [right of=yo2] {$\{y,z\},\emptyset$};
     \path[->] (1) edge node [left] {$b$} (2);
     \path[->] (1) edge node [above] {$a$} (3);
     \path[->] (2) edge node [above] {$a$} (3);
     \path[->] (2) edge [bend right=15] node [pos=0.5,below] {$b$} (4);
     \path[->] (3) edge [bend right=15] node [pos=0.5,below] {$a$} (5);
     \path[->] (3) edge node {$b$} (4);
     \path[->] (4) edge [loop below] node [below] {$b$} (4);
     \path[->] (4) edge [bend right=15] node [pos=0.5,below] {$a$} (2);
     \path[->] (5) edge [bend right=15] node [pos=0.5,above] {$a$} (3);
     \path[->] (5) edge [bend left=15] node [pos=0.5,below] {$b$} (4);
     \end{tikzpicture}
\end{tiny}
\vspace{10pt}
\end{minipage}
The $a$-transition from the accepting macro-state $(\{y\},\{z\})$ in $A'$ leads
to the non-accepting macro-state $(\{y,z\},\emptyset)$ and not to $(\{y\},\{z\})$;
the tracked set of accepting states is then reset to $\{z\}$ after a further
$a$- or $b$-transition.
This reflects the fact that no run of $A$ can stay in the accepting state
$z$ by reading the letter $a$ so that all words that contain $a$ infinitely
often are rejected by both $A$ and $A'$.
\end{example}


\begin{lemma}[Permutation method~\cite{EsparzaKRS17,HausmannEA18}]\label{lem:perm}
Let $A$ be a limit-deterministic \buchi automaton with $n$ states.
Then there is a deterministic parity automaton $A'$
with $n'$ states and $2n$ priorities such that $L(A)=L(A')$ and  $n'\leq e(n+1)!$.
\end{lemma}
\begin{proof}
We sketch just the construction of $A'$ and refer to~\cite{EsparzaKRS17,HausmannEA18}
for the proof of equivalence of $A$ and $A'$.
Intuively, $A'$ is similar to the powerset automaton of $A$, but additionally
keeps a permutation on the deterministic states in macro-states, indicating
the order in which runs leading to the respective states have last seen an accepting state.
Additionaly, states in $A'$ contain a third component which indicates the
leftmost position in the permutation that is \emph{active} or \emph{ending}
by the transitions leading to the current state in $A'$. Here, a position is
active in an $a$-transition, if the state at this position in the current
permutation is accepting; a position is said to be ending if all runs of $A$
that are represented by the state at this position end
when reading the letter $a$ or lead to a state at an older position.
A parity condition then uses this information
to detect a position in the permutation components that is active
infinitely often but, from some point
on, never ends. This ensures the existence of a continuous run of $A$ that
visits some accepting state infinitely often.\medskip

Formally, we proceed as follows. Given a limit-deterministic
\buchi automaton $A=(\Sigma,Q,q_0,\delta,F)$
with sets $Q_D,Q_N\subseteq Q$ of deterministic and nondeterministic states,
respectively, we have that every state reachable from $F$ is contained in $Q_D$.
We assume without loss of generality that $q_0\in Q_N$.
We let $\mathsf{perm}(Q_D)$ denote the set of partial permutations over $Q_D$,
that is, $\mathsf{perm}(Q_D)$ consists of all partial functions
$f:Q_D\rightharpoonup |Q_D|$ such
that $f(q)\neq f(q')$ for all $q,q'\in\mathsf{dom}(f)$ such that $q\neq q'$.
We denote the empty permutation by $[]$ ($\mathsf{dom}([])=\emptyset$).
Then we define the deterministic parity automaton
$A'=(\Sigma,Q',u_0,\delta',\Omega)$ by putting
\begin{align*}
Q'&= 2^{Q_N}\times \mathsf{perm}(Q_D)\times\{1,\ldots, 2|Q_D|+1\}&
\Omega(U,f,p)&=p
\end{align*}
and $u_0=(\{q_0\},[],1)$.
The claimed bounds on the size and number of priorities of $A'$ follows.
The transition function $\delta'$ is defined by putting, for $(U,f,p)\in Q'$ and
$a\in \Sigma$,
\begin{align*}
\delta'((U,f,p),a)=(\delta(U,a)\cap Q_N,f',p'),
\end{align*}
where $f'$ denotes the partial permutation that
is obtained by applying $a$-transitions from $\delta$ to the partial
permutation $f$, keeping the ordering
intact but removing elements that do not have an outgoing $a$-transition;
here it is crucial that all states in $f$ are deterministic so that
it is never the case that additional elements are inserted between any two elements
of the permutation.
Furthermore, we add all states from $\delta(U,a)\cap Q_D$ that do not already occur in
this new permutation to the end of it (the order of these elements is irrelevant).
Let $i\geq 1$ be the leftmost position
in $f'$ such that
$f(i)$ is defined and $\delta(f(i),a)\neq f'(i)$ (including the case that
$f'(i)$ is undefined), or we have $f'(i)\in F$.
Thus $i$ identifies the leftmost position in the partial permutation
that is active or ending (possibly both).
If no such $i$ exists, put $p'=1$.
Otherwise, if $\delta(f(i),a)\neq f'(i)$, then put $p'=2(|Q_D|-i)+3$;
if $\delta(f(i),a)=f(i)\in F$, then put $p'=2(|Q_D|-i)+2$.
\qed
\end{proof}

\begin{example}
For the limit-deterministic \buchi automaton $A$ with
$Q_N=\{x,z\}$ and $Q_D=\{y,u\}$ depicted below,
we obtain the equivalent deterministic parity automaton $A'$ using
the construction from Lemma~\ref{lem:perm}.
For brevity, we depict $A'$ with \emph{edge priorities}, thus moving
the priority component $p$ of macro-states $(U,f,p)$ to the edges.
\medskip

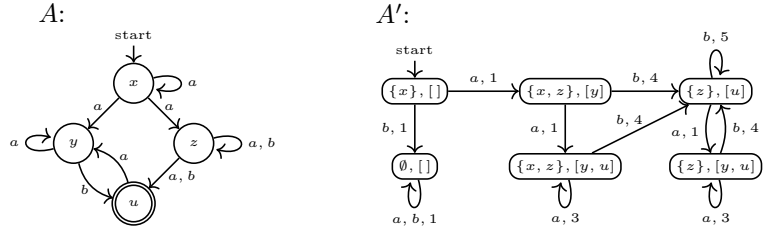
\begin{minipage}{.3\linewidth}
$\quad$
		 $A$:\\
\tikzset{every state/.style={minimum size=15pt}}
\begin{tiny}
  \begin{tikzpicture}[
		auto,
    node distance=0.8cm,
    semithick
    ]
     \node[state,initial above] (0) {$x$};
     \node (yo) [below of=0] {};
     \node[state] (1) [left of=yo] {$y$};
     \node[state] (2) [right of=yo] {$z$};
     \node[state,accepting] (3) [below of=yo] {$u$};
     \path[->] (0) edge [loop right] node [right] {$a$} (0);
     \path[->] (0) edge node [pos=0.3,left] {$a$} (1);
     \path[->] (0) edge node [pos=0.3,right] {$a$} (2);
     \path[->] (1) edge [loop left] node [left] {$a$} (1);
     \path[->] (1) edge [bend right=30] node [pos=0.3,below] {$b$} (3);
     \path[->] (2) edge [loop right] node [right] {$a,b$} (2);
     \path[->] (2) edge node [pos=0.6,right] {$a,b$} (3);
     \path[->] (3) edge [bend right=30] node [pos=0.6,right] {$a$} (1);

  \end{tikzpicture}
\end{tiny}

    \end{minipage}%
		$\quad$
		$\quad$
		$\quad$
    \begin{minipage}{.7\linewidth}
$A'$:
\tikzset{every state/.style={minimum size=15pt}}
\begin{tiny}

  \begin{tikzpicture}[
		auto,
    node distance=1.0cm,
    semithick
    ]
     \node[rounded corners, draw,initial above] (1) {$\{x\},[\,]$};
		 \node (yo) [right of=1] {};
     \node[rounded corners, draw] (2) [below of=1] {$\emptyset,[\,]$};
     \node[rounded corners, draw] (3) [right of=yo] {$\{x,z\},[y]$};
     \node[rounded corners, draw] (4) [below of=3] {$\{x,z\},[y,u]$};
		 \node (yo2) [right of=3] {};
     \node[rounded corners, draw] (5) [right of=yo2] {$\{z\},[u]$};
     \node[rounded corners, draw] (6) [below of=5] {$\{z\},[y,u]$};
		 \path[->] (1) edge node [left] {$b,1$} (2);
     \path[->] (1) edge node [above] {$a,1$} (3);
     \path[->] (2) edge [loop below] node [below] {$a,b,1$} (2);
     \path[->] (3) edge node [left] {$a,1$} (4);
     \path[->] (3) edge node {$b,4$} (5);
     \path[->] (4) edge [loop below] node [below] {$a,3$} (4);
     \path[->] (4) edge node [pos=0.4,above] {$b,4\,$} (5);
     \path[->] (5) edge [loop above] node [above] {$b,5$} (5);
     \path[->] (5) edge [bend right=20] node [pos=0.6,left] {$a,1$} (6);
     \path[->] (6) edge [loop below] node [below] {$a,3$} (6);
     \path[->] (6) edge [bend right=20] node [pos=0.5, right] {$b,4$} (5);
     \end{tikzpicture}
\end{tiny}
\vspace{10pt}
\end{minipage}
Let $\delta$ be the transition relation of $A$.
In $A'$ there is an $a$-transition with priority $1$ from the initial state
$(\{x\},[])$ to $(\{x,z\},[y])$. This is the case since $\delta(\{x\},a)=\{x,y,z\}$
so that $\{x,y,z\}\cap Q_N=\{x,z\}$. Since $y\in \delta(\{x\},a)\cap Q_D$, we add
it to the end of the permutation component which thereby changes from $[]$ to $[y]$.
We have $y\notin F$ so that there is no position in the permutation
that ends or is active. Thus the priority of this transition is $1$.
There is an $a$-loop with priority $3$ at $(\{x,z\},[y,u])$.
This is the case since $\delta(\{x,z\},a)=\{x,y,z,u\}$ so that
 $\delta(\{x,z\},a)\cap Q_N=\{x,z\}$. Also we update the permutation component
 $[y,u]$ according to reading the letter $a$: We have
 $\delta(y,a)=y$ and $\delta(u,a)=y$ and hence obtain a temporary permutation
 $[y]$. Now $\delta(\{x,z\},a)\cap Q_D$ contains the state $u$ that is appended
 to the permutation, resulting in $[y,u]$ as new permutation component.
 As $y\notin F$ and $\delta(y,a)=y$, the leftmost position in the permutation
 component is neither active nor ending. We also have $u\in F$ and $\delta(u,a)\neq u$
 so that position $2$ is both ending and active. Hence the priority of this
 transition is $2(|Q_D|-i)+3=2(2-2)+3=3$. This reflects the fact that even though
 an accepting state can be reached from $\{x,y,z,u\}$ by an $a$-transition
 in $A$ (as $u\in \delta(z,a)$), all runs that have visited an
 accepting state at least once before are residing in the state $y$ after reading $a$.
 Thus it is not possible to construct a continuous run
 that only reads the letter $a$ and still visits $u$ more than once.
 Intuitively, reading the letter $a$ merges all runs leading to $y$ or $u$, so that
 both positions in the permutation component $[y,u]$ are merged into the new first
 position containing just $y$; the second position thus is ending.
 The deterministic state $u$ to which there is
 no $a$-transition from $y$ or $u$ then is appended as new (accepting) position to the
 permutation.
\end{example}


\begin{lemma}[Safra-Piterman~\cite{Safra88,Piterman07}]\label{lem:sp}
Let $A=(\Sigma,Q,q_0,\delta,F)$ be a \buchi automaton.
Then there is a deterministic parity automaton $A'=(\Sigma,Q',u_0,\delta',\Omega)$
such that $L(A)=L(A')$, $|Q'|\in
{\mathcal{O}}(|Q|!^2)$ and $A'$ has
at most $2|Q|$ priorities.\footnote{The tight complexity analysis of
the construction in \cite{Piterman07} is in \cite{Schewe09}.}
\end{lemma}

\begin{lemma}[Parity
Determinization~\cite{ScheweVarghese14}]\label{lem:schewevarghese}
	Let $A=(\Sigma,Q,q_0,\delta,\Omega)$ be a parity automaton of rank
	$k$.
	Then there is a deterministic parity automaton
	$A'=(\Sigma,Q',u_0,\delta',\Omega')$
	such that $L(A)=L(A')$, $|Q'|\in
	{\mathcal{O}}(|Q|!^{2k})$ and $A'$ has
	at most $2|Q|k$ priorities.
\end{lemma}

\subsection{History-determinizing Word Automata}

Next we give specialized history determinization constructions for limit-deterministic
\cbuchi automata and for general \buchi automata.

\begin{lemma}[History-determinizing by focusing]\label{lem:focus}
Let $A=(\Sigma,Q,q_0,\delta,F)$ be a limit-deterministic \cbuchi word automaton. Then
there is a history-deterministic \cbuchi word automaton $A'=(\Sigma,Q',u_0,\delta',F')$ such that
$L(A)=L(A')$ and $|Q'|\leq (|F|+1)\cdot 2^{|Q|}$.
\end{lemma}

\begin{proof}
Inituively, the determinization procedure is similar to the powerset construction
but uses the limited nondeterminism that is allowed in history-deterministic
automata to guess a run that eventually stays in $F$ forever.
Information about the guessed runs is kept by annotating macro-states with
a \emph{focus}, that is, the state in which the run currently resides.
If the guess turns out to be wrong and the run leaves $F$, a new guess is taken
(a refocusing step takes place). The resulting automaton then can be shown
to be history-deterministic by using a resolver function that refocuses in a fair
manner, guaranteeing that no run is overlooked.

Formally, we put $Q''=\{(U,q)\in 2^Q\times F\mid q\in U\}$ and
\begin{align*}
Q'&=2^Q \cup Q'' &
u_0&=\{q_0\} &
F'&=Q'',
\end{align*}
from which the claimed bound on the size of $A'$ follows since
\begin{align*}
|Q'|=|2^Q \cup Q''|=
2^{|Q|}+(2^{|Q|}\cdot|F|)=(|F|+1)\cdot 2^{|Q|}.
\end{align*}
The transition relation $\delta'$ is defined by putting, for $a\in\Sigma$ and $U\subseteq Q$,
\begin{align*}
\delta'(U,a)=\{\delta(U,a)\}\cup \{(\delta(U,a),q')\mid q'\in\delta(U,a)\cap F\}
\end{align*}
and, for $a\in\Sigma$ and $(U,q)\in Q''$,
\begin{align*}
\delta'((U,q),a)=
\begin{cases}
\{(\delta(U,a),q')\} & \text{if }\delta(q,a)\cap F=\{q'\}\\
\{\delta(U,a)\} & \text{if }\delta(q,a)\cap F=\emptyset,
\end{cases}
\end{align*}
noting that since $A$ is limit-deterministic, we have $|\delta(q,a)\cap F|\leq 1$ if $q\in F$,
so that the case distinction above is exhaustive.
Given $(U,q)\in Q''$, we refer to $q$ as the \emph{focus} and for $U\in Q'$
we say that the focus is \emph{finished} at $U$. Outgoing transitions from $U\subseteq Q$
to $(U',q)\in Q''$ are \emph{refocusing transitions}.
\begin{toappendix}

	It remains
to show that $L(A)=L(A')$ and that $A'$ is history-deterministic. For the first item,
let $w=a_0a_1\ldots\in L(A)$ and let $\tau=q_0,q_1,\ldots$ be an accepting run of $A$ on $w$.
Let $i$ be a position such that $q_j\in F$ for all $j\geq i$; such $i$ exists since $\tau$ is accepting.
Construct a run $\pi$ of $A'$ as follows: Let $\pi_1=\{q_0\},T_1,\ldots,T_i$
(where $T_{i'}\in Q'$ for $1\leq i'\leq i$)
be a run of $A'$ on the first $i$ letters of $w$. Continue the run $\pi_1$ deterministically,
and let $j$ be the first position such that $T_j\subseteq Q$.
If no such $j$ exists, then the run $\pi$ that
is obtained by deterministically continuing $\pi_1$ along $w$ is an accepting run ($\pi_1$
can be continued deterministically since all states $T_j\in Q''$ are deterministic states and
no states $T_j\notin Q''$ is reached by assumption).
Otherwise, let $\pi_2=\{q_0\},T_1,\ldots,T_j$ be the deterministic continuation of $\pi_1$
up to the position $j$. We have $q_{j+1}\in T_{j+1}=\delta(T_j,a_j)$.
We extend $\pi_2$ with the transition
\begin{align*}
T_j\stackrel{a_j} {\to}((\delta(T_j),a_j),q_{j+1})
\end{align*}
which is a transition in $\delta'$.
Since the focus $q_{j+1}$ is never finished by assumption,
the nondeterministic continuation of $\pi_2$ is an accepting run.

For the converse direction, let $w=a_0a_1\ldots\in L(A')$ and let
\begin{align*}
\tau=\{q_0\},T_1,T_2\ldots
\end{align*}
be an accepting run of $A'$ on $w$.
We construct a an accepting run $\pi$ of $A$ as follows.
Let $i$ be the first position such that $T_j=(U_{j},q_{j})\in Q''$ for all $j\geq i$;
such $j$ exists since $\tau$ is accepting.
Then there is a run $\pi_1$ of $A$ on the first $i$ letters of $w$ that ends
in $q_i$. Continue $\pi_1$ with the sequence $q_{j+1},q_{j+2},\ldots$.
This results in a run of $A$ on $w$ since $q_{j+1}\in\delta(q_j,a_j)$ for all $j\geq i$.
Furthermore, the resulting run $\pi$ is accepting since $q_j\in F$
and for all $j\geq i$.

To see that $A'$ is history-deterministic, we have to define a suitable resolver.
Let $w=a_0,a_1,\ldots\in L(A')$. For any partial run
$\{q_0\},T_1,\ldots,T_m$ of $L(A')$ such that $T_m\subseteq Q$,
define $\sigma(\{q_0\},T_1,\ldots,T_m,a_m)=((\delta(T_m,a_m),q_{m+1})$, where
$q_{m+1}\in \delta(T_m,a_m)\cap F$ is a state from $\delta(T_m,a_m)\cap F$ with maximal age.
Here, the age of a state $q\in\delta(T_m,a_m)$ is the least position $j$, such that
$T_j=U_j$ or $T_j=(U_j,q_j)$ and
$U_j$
contains some state $q'$ such that
$q\in\delta(q',a_{j+1},\ldots, a_m)$.
If $\delta(T_m,a_m)\cap F=\emptyset$, then put $\sigma(\{q_0\},T_1,\ldots,T_m,a_m)=\delta(T_m,a_m)$.
Hence the resulting function $\sigma$ picks a run $\sigma(w)$ for each word $w$.
If $w\in L(A')$, then $\sigma(w)$ is an accepting run since then there
is a run in $A$ that eventually stays in $F$ forever and since
$\sigma(w)$ refocuses in a fair manner (prefering older traces over younger
traces), it is guaranteed to eventually pick a focus that is never
finished.
\hfill\qed

\end{toappendix}

\end{proof}
We note that the automaton $A'$
in the above construction is not a weak automaton, even if $A$ is a weak automaton.

\begin{example}
Consider the limit-deterministic \cbuchi automaton $A$
depicted below,
and the equivalent history-deterministic \cbuchi automaton $A'$ obtained by
using the construction from Lemma~\ref{lem:focus};
both automata accept exactly the infinite words over $\Sigma=\{a,b\}$
that contain either $a$ or $b$ finitely often.
For brevity, we depict only the reachable part of $A'$ and label macro-states
$(U,q)$ with the elements of $U$ with focus $q$ underlined. Every macro-state
in $A'$ that has a focus is accepting.\medskip

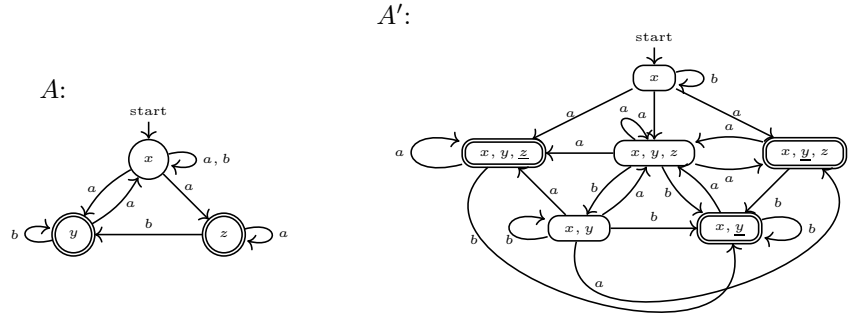
\begin{minipage}{.3\linewidth}
$\quad$
		 $A$:\\
\tikzset{every state/.style={minimum size=15pt}}
\begin{tiny}
  \begin{tikzpicture}[
		auto,
    node distance=1.0cm,
    semithick
    ]
     \node[state,initial above] (0) {$x$};
     \node (yo) [below of=0] {};
     \node[state,accepting] (1) [left of=yo] {$y$};
     \node[state,accepting] (2) [right of=yo] {$z$};
     \path[->] (0) edge [loop right] node [right] {$a,b$} (0);
     \path[->] (0) edge [bend right=15] node [pos=0.5,left] {$a$} (1);
     \path[->] (1) edge [bend right=15] node [pos=0.5,right] {$a$} (0);
     \path[->] (0) edge node [pos=0.3,right] {$a$} (2);
     \path[->] (1) edge [loop left] node [left] {$b$} (1);
     \path[->] (2) edge node [pos=0.5,above] {$b$} (1);
     \path[->] (2) edge [loop right] node [right] {$a$} (2);
       \end{tikzpicture}
\end{tiny}

    \end{minipage}%
		$\quad$
		$\quad$
		$\quad$
    \begin{minipage}{.7\linewidth}
$A'$:
\tikzset{every state/.style={minimum size=25pt}}
\begin{tiny}

  \begin{tikzpicture}[
		auto,
    node distance=1.0cm,
    semithick
    ]
     \node[rounded corners, draw,initial above] (1) {\phantom{\{}$x$\phantom{\}}};
     \node[rounded corners, draw] (3) [below of=1] {\phantom{\{}$x,y,z$\phantom{\}}};
		 \node (yo) [left of=3] {};
     \node[rounded corners, draw,accepting] (2) [left of=yo] {\phantom{\{}$x,y,\underline{z}$\phantom{\}}};
     \node[rounded corners, draw] (4) [below of=yo] {\phantom{\{}$x,y$\phantom{\}}};
		 \node (yo2) [right of=3] {};
     \node[rounded corners, draw,accepting] (5) [right of=yo2] {\phantom{\{}$x,\underline{y},z$\phantom{\}}};
     \node[rounded corners, draw,accepting] (6) [below of=yo2] {\phantom{\{}$x,\underline{y}$\phantom{\}}};
     \path[->] (1) edge [loop right] node [right] {$b$} (1);
     \path[->] (1) edge node [left] {$a$} (2);
     \path[->] (1) edge node [left] {$a$} (3);
     \path[->] (1) edge node [right] {$a$} (5);
     \path[->] (2) edge [loop left] node [left] {$a$} (2);
     \path[->] (2) edge [in=120,out=60,bend right=120,distance=2cm] node [pos=0.2,below] {$b$} (6);
     \path[->] (3) edge [loop above,out=145,in=115,distance=0.5cm] node [above] {$a$} (3);
     \path[->] (3) edge node [above] {$a$} (2);
     \path[->] (3) edge [bend right=15] node [pos=0.5,below] {$a$} (5);
     \path[->] (3) edge [bend right=15] node [pos=0.5,left] {$b$} (4);
     \path[->] (4) edge [loop left] node [below] {$b$} (4);
     \path[->] (4) edge node [right] {$a$} (2);
     \path[->] (4) edge [bend right=15] node [pos=0.5,right] {$a$} (3);
     \path[->] (4) edge [in=120,out=60,bend right=120,distance=1.8cm] node [pos=0.2,above] {$a$} (5);
     \path[->] (4) edge node {$b$} (6);
     \path[->] (5) edge [bend right=15] node [pos=0.5,above] {$a$} (3);
     \path[->] (5) edge node {$b$} (6);
     \path[->] (6) edge [loop right] node [right] {$b$} (6);
     \path[->] (3) edge [bend right=15] node [pos=0.5,left] {$b$} (6);
     \path[->] (6) edge [bend right=15] node [pos=0.5,right] {$a$} (3);
     \end{tikzpicture}
\end{tiny}
\vspace{-10pt}
\end{minipage}
Let $\delta$ be the transition relation of $A$.
Then we have $a$-transitions from $x$ to $x,y,\underline{z}$, to
$x,y,z$ and to $x,\underline{y},z$. This is the case since
$\delta(x,a)=\{x,y,z\}$. Since both $y$ and $z$ are accepting states in $A$,
$A'$ has, when reading $a$ at the state $x$,
the history-deterministic choice to focus either $y$ or $z$ (or none
of the two). On the other hand, we have e.g. an $a$-transition from $x,\underline{y}$ to
the non-accepting macro-state $x,y,z$ since $\delta(y,a)=x$
is not an accepting state. Hence the $a$-transition from $x,\underline{y}$ finishes the
focus and another refocusing step is necessary in order for a run to be accepting.
Thus $A'$ accepts for instance the word $(aba)b^\omega$ (also accepted by $A$)
by staying unfocused
when reading $abab$, leading to the partial run $x\stackrel{a}{\to} x,y,z
\stackrel{b}{\to} x,y \stackrel{a}{\to} x,y,z \stackrel{b}{\to} x,y$; then
the automaton can focus on
$y$, continuing the run with $x,y \stackrel{b}{\to} x,\underline{y}
\stackrel{b}{\to}  x,\underline{y}\ldots$, resulting in an overall accepting run.
For the word $a^\omega$ however,
there is a non-accepting run
$x\stackrel{a}{\to} x,y,z
\stackrel{a}{\to} x,\underline{y},z \stackrel{a}{\to} x,y,z \stackrel{a}{\to} x,\underline{y},z$ in which $y$ is focused infinitely often, but also finished infinitely often.
This shows that fair order of focusing is crucial in resolving the history-determinism:
Every run in which the automaton eventually focuses the state $z$ is accepting.
\end{example}

\begin{lemma}[Henzinger-Piterman~\cite{HenzingerP06}]\label{lem:hp}
Let $A$ be a nondeterministic B\"uchi word automaton with $n$ states. Then
there is a history-deterministic parity word automaton $A'$ with $n'$ states
such that $L(A)=L(A')$ and
$n'\in\mathcal{O}(3^{n^2})$.
\end{lemma}

Notice that the size of a history-deterministic automaton, in the
general case, is larger than the size of the deterministic automaton.
The potential advantage of using history determinization would be to
have a simpler structure of the resulting automaton.

\subsection{Application to Emptiness Checking and $\mu$-Calculus
Satisfiability}
To conclude this section, we state the connection between the structure
of the alternating parity tree automaton $A$ and the structure of the
tracking automaton $T_A$.

\begin{lemma}
	\begin{itemize}
		\item If $A$ is an alternating weak tree automaton, then the
		tracking
		automaton $T_A$ is a weak word automaton.
		\item If $A$ is limit deterministic, then the tracking
		automaton $T_A$ is limit deterministic.
	\end{itemize}
\end{lemma}

\begin{proof}
\begin{itemize}
\item
	Let $A$ be a weak automaton. Then all strongly connected components in
	$T_A$ either contain only states with priority $2$ or only states
	with priority $1$.
	Regarding states with priority $1$ as non-accepting and states with
	priority $2$ as
	accepting, $T_A$ can be seen as a weak automaton (and as a \cbuchi
	automaton).
\item
    Recall that the transition relation of $T_A$ is $\Gamma$ and the priority function
    of $T_A$ is $\overline{\Omega}$.
    Let $A$ be limit deterministic, let $q\in Q$ such that $\overline{\Omega}(q)=p$ is even,
    let $q'\in C_q$ and let $a\in\Sigma$.
    We have to show that $|\Gamma(q',a)\cap C_q|\leq 1$. Since $q'\in C_q$,
    $q'$ is reachable from $q$ by a path that visits states with priority at most
    $p$.
    The only case where we have $|\Gamma(q',a)|>1$ is when $q'\in Q_\wedge$.
    Also $\Omega(q)=\overline{\Omega}(q)-1=p-1$
    is odd and $q'$ is reachable in $A$ by a path that visits nodes with
    priority at most $p-1$. Since $A$ is limit-deterministic,
    we have $|\delta(q',a)\cap Q_{\leq p-1}|\leq 1$
    which implies $|\Gamma(q',a)\cap C_q|\leq 1$ since $C_q\subseteq Q_{\leq p-1}$
    by definition of compartments and since $\Gamma(q',a)\subseteq
    \delta(q',a)$ by definition of $\Gamma$.
\end{itemize}
	\qed
\end{proof}

\begin{remark}
    If $A$ is limit linear, then $A$ is weak. By the above lemma, $T_A$
	is a weak automaton with $F$ being the states with priority $2$.
	Given $q\in F$ so that $\Omega(q)=\overline{\Omega}(q)-1=1$, there is
	exactly one path from $q$ to $q$ in $A$, since $A$ is limit linear. Except for the self-loops introduced by
	non-manipulating transitions in $T_A$, there is exactly one path from $q$ to $q$
	in $T_A$ that stays in $F$. We note that the concept of limit-linear word automata can
	be slightly extended to accommodate self-loops,
	using a notion of \emph{synchronizing}
	transitions in such a way that the method from~\autoref{lem:limitlinear}
	can be employed, obtaining the same complexity result.
	For brevity, we omit the technical details here and
	refer to~\cite{Hausmann18} instead.
\end{remark}

By using the bespoke determinization and history-determinization
constructions stated above we achieve below better complexity bounds.

\begin{corollary}\label{cor:treeemptiness}
Let $A$ be an alternating parity tree automaton reading an alphabet of size $m$,
with $n$ states and rank $k$.
Depending on the structure of $A$, the complexity of emptiness checking for $A$ is as follows (where $s=(nm+1)\cdot 2^n$).

\begin{tabbing}
\hspace*{5pt}\= \hspace*{19em} \= \kill 
    \> -- If $A$ is limit-linear: \> \paritytime{s \cdot n^2\cdot 2^n}{2} \\
    \> -- If $A$ is limit-deterministic and weak:  \> \paritytime{s\cdot n\cdot 2^n}{2} \\
    \> -- If $A$ is weak: \> \paritytime{s \cdot 3^n}{2} \\
    \> -- If $A$ is limit-deterministic: \> \paritytime{s\cdot e\cdot(nk)!}{2nk} \\
    \> -- In any case: \> \paritytime{s \cdot
    \mathcal{O}((n)!^{2k})}{2nk}
\end{tabbing}

\end{corollary}

\begin{remark}
Given a formula $\varphi$, the automaton $A(\varphi)$ makes very limited use of the alphabet $\Sigma$
(in fact, it is only used to check for satisfaction of propositional atoms).
For satisfiability checking, the guessing and memorizing
of letters in the emptiness game $G^*_A$ can hence be avoided by letting
\pbox immediately win all nodes whose state component $s\subseteq Q$ in the
strategy arena contains  $\{p,\neg p\}$ for some atom $p$. Furthermore,
the state component $s\subseteq Q$ of nodes in the strategy arena is always
contained in the label of states of the (history) deterministic variant
of the tracking automaton $H_A$.
Hence $G^*_A$ can be slightly
adapted to obtain the following complexity bounds, matching previously
known results for guarded formulas.
\end{remark}

\begin{corollary}\label{cor:resultsmu}
Let $\varphi$ be a $\mu$-calculus formula and let $n=|\FL(\varphi)|$,
$k=\mathsf{ad}(\varphi)$. Then the time complexity of deciding satisfiability of $\varphi$
is as follows.

\begin{tabbing}
\hspace*{5pt}\= \hspace*{19em} \= \kill 
    \> -- If $\varphi$ is limit-linear: \> \paritytime{n^2\cdot 2^n}{2} \\
    \> -- If $\varphi$ is aconjunctive and alternation-free:  \> \paritytime{n\cdot 2^n}{2} \\
    \> -- If $\varphi$ is alternation-free: \> \paritytime{3^n}{2} \\
    \> -- If $\varphi$ is aconjunctive: \> \paritytime{e\cdot(nk)!}{2nk} \\
    \> -- In any case: \> \paritytime{\mathcal{O}((n)!^{2k})}{2nk}
\end{tabbing}

Let $\varphi$ be satisfiable. Then $\varphi$ has a model of size $2^{O(nk)\log n}$,
and of size $3^n$ if $\varphi$ is alternation-free.
\end{corollary}

\begin{remark}

In~\cite{FriedmannLange13a}, the authors present
a tableaux-based satisfiability algorithm for unguarded formulas;
unguardedness is handled by an auxilliary tableau rule and by extending the tracking automaton with an additional priority to detect \emph{inactive} traces. Using this approach however,
the tracking automaton for unguarded alternation-free formulas is
(in contrast to our framework) \emph{not} a \cbuchi automaton
and \cbuchi methods for (history)-determinization can not be used to obtain \buchi games
that characterize satisfiability.

Our treatment of aconjunctive and alternation-free formulas employs a focusing method
(\autoref{lem:focus}) to history-determinize limit-deterministic \cbuchi automata.
This generalizes the focus games for CTL~\cite{LangeStirling01} to
the aconjunctive alternation-free $\mu$-calculus and sheds light on the automata theoretic
background of focus games.
\end{remark}

\section{Conclusions}
\label{section:conclusions}

We surveyed the approach to deciding the satisfiability of the modal
$\mu$-calculus through a reduction to alternating parity tree automata
emptiness.
We present the solution to the emptiness of alternating parity tree
automata as a combination of a structural game construction with word
automata for defining the winning condition.
Interestingly the structural game construction remains fixed regardless
of the exact structure of the automaton.
The exact structure, however, greatly affects the properties of the
word automata for the winning condition.
This, in turn, can be exploited to give improved complexity results for
various fragments of the $\mu$-calculus by concentrating on bespoke
word automata conversion constructions.

\bibliographystyle{splncs04}
\bibliography{lib}

\end{document}